\PassOptionsToPackage{dvipsnames}{xcolor}
\documentclass[pdflatex,sn-mathphys-num]{sn-jnl}

\usepackage{graphicx}%
\usepackage{multirow}%
\usepackage{amsmath,amssymb,amsfonts}%
\usepackage{amsthm}%
\usepackage{mathrsfs}%
\usepackage[title]{appendix}%
\usepackage{textcomp}%
\usepackage{manyfoot}%
\usepackage{booktabs}%
\usepackage{algorithm}%
\usepackage{algorithmicx}%
\usepackage{algpseudocode}%
\usepackage{listings}%
\usepackage{tcolorbox}%
\usepackage{esint}

\hypersetup{
  colorlinks=true,
  citecolor=OliveGreen,
  linkcolor=Orange,
  urlcolor=MidnightBlue
}

\newtheoremstyle{thmstyleone}
{18pt plus2pt minus1pt}
{18pt plus2pt minus1pt}
{\itshape}
{0pt}
{\bfseries}
{}
{.5em}
{}
\theoremstyle{thmstyleone}

\newtheorem{theorem}{Theorem}[section]

\newtheorem{lemma}[theorem]{Lemma}
\newtheorem{proposition}[theorem]{Proposition}

\newtheorem{example}[theorem]{Example}

\newtheorem{remark}[theorem]{Remark}

\newtheorem{assumption}[theorem]{Assumption}
\numberwithin{equation}{section}

\newcommand{\Sp}{\mathrm{S}^{1}}
\newcommand{\R}{\mathbb{R}}
\newcommand{\C}{\mathbb{C}}
\newcommand{\CP}{\mathbb{CP}^{1}}

\newcommand{\Cc}{\mathrm{C}}

\begin{document}

\title[Average Winding Number]{Average Winding Number for Determinantal Curves associated with 2-Matrix Models in the Class AIII}

\author{\fnm{Mathieu} \sur{Yahiaoui}}\email{m.yahiaoui@unimelb.edu.au}

\author{\fnm{Mario} \sur{Kieburg}}\email{m.kieburg@unimelb.edu.au}

\affil[1]{\orgdiv{Department of Mathematics and Statistics}, \orgname{The University of Melbourne}, \orgaddress{\street{813 Swanston Street}, \city{Parkville}, \postcode{3052}, \state{Victoria}, \country{Australia}}}

\abstract{To classify one-dimensional disordered quantum systems with chiral symmetry, we analyse the winding number of the determinant of a parametrized non-Hermitian random matrix field over the unit circle modelling the off-diagonal block of a disordered chiral Hamiltonian. The associated partition function is computed explicitly for a broad class of additive two-matrix models extending beyond the Ginibre Unitary Ensemble. In the large-dimension limit, we derive an asymptotic expansion of the average winding number whose leading term exhibits universal features, up to the tail behaviour of the underlying random matrix ensemble, and identify a new correction term absent in the previously studied Ginibre case.}




\maketitle

\section{Introduction}\label{sec:intro}

The discovery of the quantum Hall effect in 1980 \cite{klitzingNewMethodHighAccuracy1980} revealed that, as the strength of an external magnetic field increases, the Hall conductance of a two-dimensional electron gas forms quantized plateaux. These plateaux persist even in the presence of impurities or lattice imperfections, indicating a remarkable robustness of the quantization to disorder. 

This phenomenon initiated the modern study of topological phases of matter, where insulating bulk behaviour coexists with conducting edge states. The systematic classification of such phases was achieved in the seminal work \cite{SchnyderClassificationTopologicalInsulators2008} and relies on the Altland-Zirnbauer classification \cite{AltlandNonstandardSymmetryClasses1997}, sometimes referred to as the \textit{Tenfold Way} \cite{BaezTenfoldWay2020}, which organises spectrally gapped Hamiltonians according to the Cartan classification of Riemannian symmetric spaces.

More recently, non-Hermitian Hamiltonians have attracted interest as effective models for open quantum systems where dissipation may occur. In that context, the interplay between topology and spectral structure plays a central role, see \cite{AshidaNonHermitianPhysics2020} for a recent review.

Through the bulk-boundary correspondence (see for example \cite{GrafBulkEdgeCorrespondenceDisordered2018,SawadaBulkBoundaryCorrespondenceErgodic2024}), edge phenomena observed in systems with open boundary conditions may be understood by studying the same system with closed boundary conditions often referred to as the \textit{bulk}. The corresponding bulk Hamiltonian has the aforementioned discrete translational lattice invariance. Its Fourier transform may be parametrized by $p\in\Sp$ on the unit circle when considering the one-dimensional setting. Furthermore, its eigenvectors form a complex vector bundle over $\Sp$. Non-triviality of the latter corresponds to the existence of protected edge states in the original system.

We consider one-dimensional quantum systems with discrete translational symmetry and study their Hamiltonians in the quasi-momentum representation. This formulation, natural for crystalline lattices, follows from Bloch's theorem which states that eigenstates of the Schr\"odinger operator can be written as
\begin{align}\label{th:bloch-d1}
    \psi_{\mathrm{k}}\left(x\right)=e^{i\mathrm{k}x}u_{\mathrm{k}}\left(x\right),
\end{align}
where $\mathrm{k}$ is the quasi-momentum varying over the Brillouin zone $[-\pi,\pi]\cong\Sp$ and $u_{\mathrm{k}}$ is periodic with respect to the lattice spacing set to unity for simplicity. In what follows, we identify $p=e^{i\mathrm{k}}\in\Sp$.

Among the ten Altland-Zirnbauer symmetry classes, the chiral \textbf{class AIII} exhibits a particularly rich structure in odd dimensions due to the presence of an energy gap at the Fermi level. It corresponds to quantum systems that are neither time-reversal invariant nor particle-hole invariant yet remain symmetric under the composition of time-reversal and particle-hole conjugation, see \cite[Section~E]{ChiuClassificationTopologicalQuantum2016}. In such models, the determinant of the off-diagonal block of the Hamiltonian defines a map from the circle $\Sp$ to the punctured complex plane $\C^{*}$ whose winding number around the origin serves as a characteristic integer invariant. 

As noted in \cite{braunWindingNumberStatistics2022,hahnUniversalCorrelations2024}, this invariant given by the determinantal curve remains meaningful under random perturbations even though averaging may smear out the spectral gap. Indeed, each realization of the random Hamiltonian retains a finite gap whose typical length scales with the local mean level spacing.

We model the system by a random field $H$ of complex chiral Hermitian matrices over the unit circle, representing a one-dimensional disordered crystal in quasi-momentum space. The randomness encodes disorder and random matrix theory provides a tractable framework for the corresponding ensemble averages. The invariant of interest is the winding number of the determinant of the off-diagonal block of the Hamiltonian. While this quantity is strictly defined for periodic Bloch Hamiltonians, disorder generally breaks periodicity. To recover a well-defined invariant one may assume that the disorder is periodic over a large yet finite window whose length is eventually taken to infinity.

Our analysis extends the computation of the partition function associated with the random winding number, previously derived for the Ginibre Unitary Ensemble (\textbf{GinUE}) additive $2$-matrix model \cite{braunWindingNumberStatistics2022,hahnWindingNumberStatistics2023,hahnUniversalCorrelations2024}, to a broader class of non-Hermitian additive $2$-matrix models for which exact formulas can still be obtained, allowing us to extract large-dimension asymptotics for the average winding number.

Section~\ref{sec:main} introduces the model and the associated constructions including the winding number, the partition function and P\'olya ensembles of multiplicative type~\cite{kieburgExactRelationSingular2016,kieburgProductsRandomMatrices2019,forsterPolynomialEnsemblesPolya2021}. In this section, we also summarise our main results: \autoref{theorem1} provides an exact expression valid for finite size $N$ and for all P\'olya $2$-matrix models whereas \autoref{theorem2} gives the asymptotic expansion of the mean winding number for the Muttalib-Borodin ensemble~\cite{borodinBiorthogonalEnsembles1998} of Laguerre type. As a comparison, \autoref{proposition.gin} confirms that the result of \cite{braunWindingNumberStatistics2022,hahnWindingNumberStatistics2023} for the GinUE additive 2-matrix model extends to more general matrix-valued Gaussian random fields. Proofs of the finite and large-$N$ results are presented in Sections~\ref{sec1} and~\ref{sec:asymp}, respectively, and the auxiliary integral identities used therein are collected in the appendices. We conclude in Section~\ref{sec13}.

\section{Preliminaries and main results}\label{sec:main}

\subsection{Winding number of a determinantal curve}\label{sec:winding-number}

As aforementioned, we aim to study a specific classification of topological phases of random Hamiltonians. More precisely, we focus our attention on class AIII in the one-dimensional setting, for which the invariant happens to be a winding number.
In a chiral basis, a class AIII Hamiltonian has the following form
\begin{align}\label{chiral.Hamiltonian}
    H\left(p\right)=\begin{pmatrix}
0_{N} & K\left(p\right) \\
K^{\dagger}\left(p\right) & 0_{N}
\end{pmatrix},
\end{align}
where $(\cdot)^{\dagger}$ denotes the Hermitian conjugate and we recall that $p$ denotes the quasi-momentum varying over the Brillouin zone, which is isomorphic to the unit circle. The matrix-valued random field $K$ is assumed to be analytic entrywise over $\mathrm{S}^{1}$.

Topology becomes relevant when assuming the existence of an energy gap around zero. The presence of a spectral gap gives rise to several topological properties of this set of Hamiltonians, which can be captured by different invariants.

In a naive approach, the determinantal curve $p\mapsto \det\left(H\left(p\right)\right)$ would not yield any significant topological information since $H$ is Hermitian and its determinant takes values in the homotopically trivial space $\R_{+}$. Following~\cite{MaffeiTopologicalCharacterizationChiral2018}, one may instead consider the determinantal curve over $\Sp$ associated with the sub-block $K$. We note that the eigenvalues of $K$ are not necessarily periodic and some might be permuted, see \autoref{fig:eigenflow}.

\begin{figure}[t!]
  \centering
  \includegraphics[width=\linewidth,height=0.42\textheight,keepaspectratio]{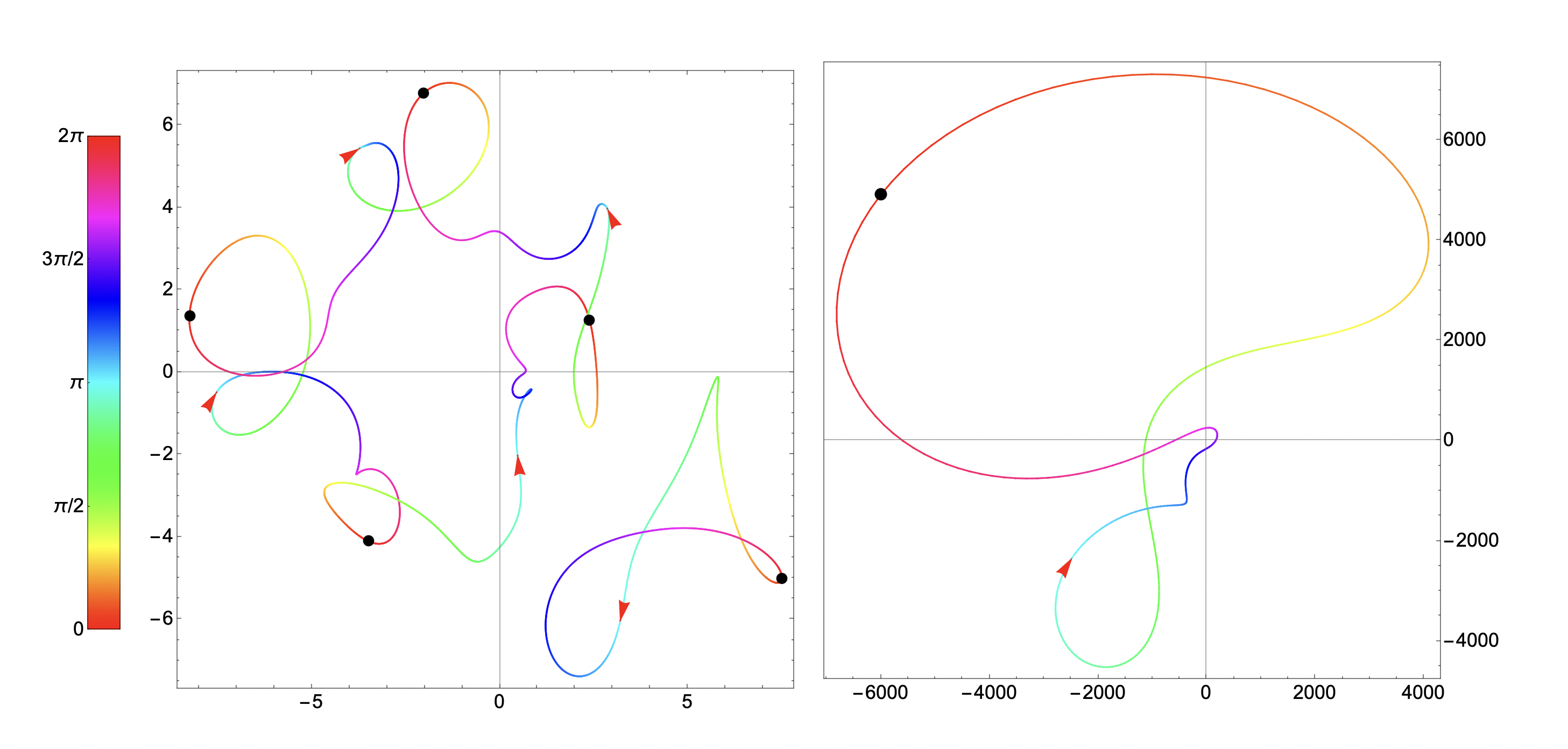}
  \caption{On the left: eigenvalue flow of $K(p)=(\frac{1}{2i}(p-p^{-1})+p)K_{1}+(3-\frac{1}{10}(p+p^{-1})+p^{-2})K_{2}$ where the two source random matrices $K_{1}$ and $K_{2}$ are independently drawn from the $\mathrm{GinUE}(5)$. On the right: its associated determinantal curve. The black dots indicate the $5$ eigenvalues of $K(0)$ and $\det(K(0))$ while the red arrows mark the points where $\mathrm{Arg}(p)=\pi$. The winding number around the origin of the determinantal curve equals $+1$. The eigenvalue flow is not $2\pi$-periodic: as $p$ winds once around $\Sp$, one eigenvalue traces a closed orbit whereas the remaining four form a $4$-cycle. As a result, the flow is $8\pi$-periodic. Plots generated with Mathematica}
  \label{fig:eigenflow}
\end{figure}

The determinant of $K$ is the simplest $2\pi$-periodic function that encodes the entire spectrum, to that end we define the determinantal curve
\begin{align}
    \mathscr{C}:p\mapsto \det\left(K\left(p\right)\right).
\end{align}
The energy gap of an individual chiral Hamiltonian~\eqref{chiral.Hamiltonian} is in direct correspondence with the spectral gap at the origin of $K$ since $\det H=|\det K|^2$.
We thus consider the corresponding winding number of $\det\left(K\left(p\right)\right)$ around the origin explicitly given by
\begin{align}\label{windingnumber1}
    \mathrm{Wind}_{N}\left(\mathscr{C},0\right)=\frac{1}{2\pi i}\oint_{\mathrm{S}^{1}}^{}w\left( p \right)\mathrm{d}p,
\end{align}
where $w$ is the winding number density defined as
\begin{align}\label{wdnumberdensity1}
    w\left( p \right)=\frac{\mathrm{d}}{\mathrm{d}p}\log\left( \det\left( K\left( p \right) \right) \right)=\frac{1}{\det\left( K\left( p \right)\right)}\frac{\mathrm{d}}{\mathrm{d}p} \det\left( K\left( p \right) \right). 
\end{align}
This quantity will be our object of interest and we will study its $m$-point correlation functions with $m\geq 1$ and $N\geq 1$.
The arguments of $\mathrm{Wind}_{N}$ shall highlight the chosen curve $\mathscr{C}$ and the point it winds around. For short we denote the winding number by $\mathrm{Wind}_{N}=\mathrm{Wind}_{N}\left(\mathscr{C},0\right)$.

To investigate the distribution of that random winding number, a natural quantity to consider is the expectation of its $m$th moment as it opens the way for a possible central limit theorem, i.e.,
\begin{align}
    \mathbb{E}\Big(\mathrm{Wind}_{N}^{m}\Big)=(2\pi i)^{-m}\oint_{\mathbb{T}^{m}}\mathbb{E}\Bigg(\prod_{k=1}^{m} w\left( p_{k} \right)\Bigg)\mathrm{d}p_{1}\cdots \mathrm{d}p_{m},\label{moment.def}
\end{align}
where $\mathbb{T}^{m}:=\Sp\times\cdots\times\mathrm{S^{1}}$ denotes the $m$-dimensional unit torus.
The expectation value inside the integral can be traced back to the partition function
\begin{align}\label{partitionfunction1}
    \mathscr{Z}^{\left(N\right)}_{m}\left ( \mathbf{p},\mathbf{q} \right )=\mathbb{E}\Bigg(\prod\limits_{k=1}^{m}\frac{\det\left ( K\left ( p_{k} \right ) \right )}{\det\left ( K\left ( q_{k} \right ) \right )}\Bigg)
\end{align}
with $\mathbf{p}=\left(p_{1},\ldots,p_{m}\right)$ and $\mathbf{q}=\left(q_{1},\ldots,q_{m}\right)$ two collections of points on $\Sp$.
The associated $m$-point correlation function 
\begin{align}
    \Cc_{m}^{(N)}(\mathbf{p})=\mathbb{E}\Bigg(\prod_{k=1}^{m} w\left( p_{k} \right)\Bigg)
\end{align}
is then obtained by taking successive derivatives of $\mathscr{Z}^{\left(N\right)}_{m}$ and setting $\mathbf{p}=\mathbf{q}$. For example the 1-point and 2-point correlation functions are given by
\begin{align}
        &\Cc^{\left ( N \right )}_{1}\left ( p \right )=\mathbb{E}\left ( w\left ( p \right ) \right )=\left.\frac{\mathrm{d}}{\mathrm{d}p} \mathscr{Z}^{\left(N\right)}_{\mathrm{1}}\left(p,q\right) \right|_{q=p},
        \label{C1.def}\\
        &\Cc^{\left ( N \right )}_{2}\left ( p_{1},p_{2} \right )=\mathbb{E}\left ( w\left ( p_{1} \right )w\left ( p_{2} \right ) \right )=\left.\frac{\mathrm{d}^{2}}{\mathrm{d}p_{1}\mathrm{d}p_{2}} \mathscr{Z}^{\left(N\right)}_{\mathrm{2}}\left(p_{1},p_{2},q_{1},q_{2}\right) \right|_{\substack{q_{1}=p_{1}\\q_{2}=p_{2}}}.
\end{align}
As shown above, the integrating $\Cc^{\left ( N \right )}_{m}$ over $\mathbb{T}^{m}$ gives the $m$th moment of $\mathrm{Wind}_{N}$.

\subsection{Additive 2-matrix models}\label{sec:random-matrix}

The 2-matrix model for which both random matrices are independently drawn from the GinUE, which is defined by the requirement that all entries are independent standard complex Gaussian random variables and was first introduced in \cite{ginibreStatisticalEnsemblesComplex1965}, has been studied in detail recently in \cite{hahnWindingNumberStatistics2023}. A recent survey may be found in \cite{byunProgressStudyGinibre2025}.

We summarise these results here briefly. To that end, we fix two complex-valued and differentiable parameter functions $a$ and $b$ over the unit circle $\Sp$ and set 
\begin{align}\label{functions1}
    \nu\left(p\right)=\begin{pmatrix}
a\left ( p \right )
\\ 
b\left ( p \right )
\end{pmatrix}\in\C^2 ,\qquad\kappa\left(p\right)=\frac{a\left ( p \right )}{b\left ( p \right )}\in\CP,\qquad\mathrm{J}=\begin{pmatrix}
        0 & 1 \\ 
        -1 & 0
        \end{pmatrix},
\end{align}
where $\CP$ denotes the complex projective line.
To construct the additive 2-matrix model, we draw two random matrices $K_{1}$ and $K_{2}$ and form 
\begin{align}\label{def:2matrixmodel}
    K(p)=a\left(p\right)K_{1}+b\left(p\right)K_{2}.
\end{align}
\begin{assumption}\label{assumption1}
    For the winding number to be well-defined we require that the two parameter functions $a$ and $b$ never vanish simultaneously over $\Sp$. Furthermore, each of them is assumed to have isolated zeroes, ensuring that their logarithmic derivatives remain well-defined and yield only simple poles over $\Sp$. This guarantees the tractability of subsequent computations.
\end{assumption}

In the special case where the source matrices are drawn from the GinUE of size $N$ the result for the partition function is quite simple and was derived in~\cite[Sec.~IV]{hahnWindingNumberStatistics2023}:
\begin{align}\label{partitiongaussian}
    \mathscr{Z}^{\left(N\right)}_{m}\left ( \mathbf{p},\mathbf{q} \right )=\frac{\det\Big(\frac{1}{\nu\left ( p_{i} \right )^{\top}J\nu\left ( q_{j} \right )}\left ( \frac{\nu\left ( q_{j} \right )^{\dagger}\nu\left ( p_{i} \right )}{\nu\left ( q_{j} \right )^{\dagger}\nu\left ( q_{j} \right )} \right )^{N}\Big)_{1\leq i,j\leq m}}{\det\Big(\frac{1}{\nu\left ( p_{i} \right )^{\top}J\nu\left ( q_{j} \right )}\Big)_{1\leq i,j\leq m}}.
\end{align}
This result highlights an elegant and algebraically rich structure for the partition function associated with the random winding number.

We aim to expand this result to a larger class of random matrix ensembles and to include complex conjugates in the partition function \eqref{partitionfunction1} in order to capture the real and imaginary parts of the winding number density \eqref{wdnumberdensity1}. This class of random matrices is called \textit{multiplicative P\'olya ensembles on} the complex general linear group $\mathrm{GL}_{N}(\C)$ and will be reviewed in the next subsection.

\subsection{P\'olya ensembles of multiplicative type on $\mathrm{GL}_{N}(\C)$}\label{sec:polya}

The cornerstone of the computation in~\cite{hahnWindingNumberStatistics2023} was to reduce the 2-matrix model to a single matrix model which consists of the product of the first matrix with the inverse of the second one. This reduction is what ultimately makes the closed-form result~\eqref{partitiongaussian} possible. To generalize this idea, we need a class of random matrices that is closed under product and inverse while retaining suitable integrable structure.

Such a class has been constructed \cite{kieburgProductsRandomMatrices2019,forsterPolynomialEnsemblesPolya2021} in the past decade and is referred to as \textit{P\'olya ensembles on} $\mathrm{GL}_{N}\left(\C\right)$. It has many convenient analytical and algebraic properties, for instance both its collections of eigenvalues and of singular values form determinantal point processes.
We will give a brief introduction, an in-depth treatment can be found in~\cite{kieburgExactRelationSingular2016,kieburgProductsRandomMatrices2019,forsterPolynomialEnsemblesPolya2021}. 
 
To define P\'olya ensembles (of multiplicative type on $\mathrm{GL}_{N}\left(\C\right)$), we introduce the following Euler-type differential operators
\begin{align}
    \Theta^{\left ( 0 \right )}= \mathrm{Id},\quad\Theta^{\left ( n \right )}=\Big(-x\frac{\mathrm{d}}{\mathrm{d}x}\Big)^{n}\quad\text{ for }n\geq 1.
\end{align}
Let $I$ be an interval of $\R$ containing $1$ and $n$ a strictly positive integer, we introduce the functional set 
\begin{align}
    \mathrm{L}^{1,n}_{I}\left ( \R_{+} \right )=&\biggl\{ f\in\mathrm{L}^{1}\left ( \R_{+} \right )\cap\mathscr{C}^{n}\left( \R_{+} \right),\,\forall \left (s,k  \right )\in I\times\left [ 0,n \right ]:\nonumber\\
    &\int_{0}^{\infty}\left | t^{s-1}\Theta^{\left ( k \right )}\left (f  \right )\left ( t \right ) \right |\mathrm{d}t<+\infty \biggl \},
\end{align}
where $\mathrm{L}^{1}\left ( \R_{+} \right )$ and $\mathscr{C}^{n}\left( \R_{+} \right)$ are respectively the Lebesgue integrable and the $n$-times continuously differentiable functions on the set of positive real numbers $\R_+$.

In what follows, $f$ is an element of $\mathrm{L}^{1}\left(\R_{+}\right)$ and $A>0$. The multiplicative P\'olya ensembles are intimately related to the Mellin transform which alongside the incomplete Mellin transform is given by
\begin{align}\label{mellintransform1}
    \mathcal{M}\left [ f \right ]:z\mapsto \int_{0}^{+\infty}t^{z-1}f\left ( t \right )\mathrm{d}t, &&\textbf{Mellin transform},
    \\
    \mathcal{M}\left [ f \right ]\left(\,\cdot\,,A\right):z\mapsto \int_{0}^{A}t^{z-1}f\left ( t \right )\mathrm{d}t,&&\textbf{incomplete Mellin transform},\label{mellintransform1.inc}
\end{align}
While the exact domain of definition of the Mellin transform in the complex plane can be quite complicated, it always contains at least the line $i\R+1:=\left \{ z\in\C,\,\mathfrak{Re}\left ( z \right )=1 \right \}$ due to the absolute integrability of $f\in\mathrm{L}^{1}\left(\R_{+}\right)$.

Related to the Mellin transform is the multiplicative convolution of two functions $f,g\in\mathrm{L}^{1}\left(\R_{+}\right)$ given by
\begin{align}\label{multiplicativeconvolution}
    f\circledast g:x\mapsto \int_{0}^{+\infty}f\left ( \frac{x}{y} \right )g\left ( y \right )\frac{\mathrm{d}y}{y}.
\end{align}
Notably the Mellin transform is injective over $\mathrm{L}^{1}\left(\R_{+}\right)$ and for a complex number $z$ where both $\mathcal{M}\left [ f \right ]$ and $\mathcal{M}\left [ g \right ]$ are defined we have
\begin{align}
    \mathcal{M}\left [ f\circledast g \right ]\left(z\right)=\mathcal{M}\left [ f\right ]\left(z\right)\mathcal{M}\left [ g\right ]\left(z\right).
\end{align}
Additionally, we have the following property
\begin{align}\label{id:mellintheta}
    \mathcal{M}\left [ \Theta^{(n)}f\right ]\left(z\right)=z^{n}\mathcal{M}\left [ f\right ]\left(z\right).
\end{align}
A more complete treatment of the properties of the Mellin transform can be found in~\cite[Sec.~2]{kieburgProductsRandomMatrices2019}.

The class of \textit{P\'olya ensembles} is based on the notion of a \textit{P\'olya frequency function} which sits at the core of these random matrix ensembles.
A measurable function in $\mathrm{L}^{1}\left ( \R \right )$ is called a \textit{P\'olya frequency function of order} $n$, denoted by $f\in\mathrm{PF}_{n}$, if for all collections of real numbers $(x_{i}),\,(y_{j})$ in $\R^{n}$ and for all integer $1\leq k\leq n$, $f$ satisfies
    \begin{align}
        \Delta_{k}\left ( x_{1},\cdots,x_{k} \right )\Delta_{k}\left ( y_{1},\cdots,y_{k} \right )\det\left ( f\left ( x_{i}-y_{j} \right ) \right )_{1\leq i,j\leq k}\geq 0,
    \end{align}
where we have used the Vandermonde determinant given by
\begin{equation}
    \Delta_{k}\left ( x_{1},\cdots,x_{k} \right )=\det\left(x_{i}^{j-1}\right)_{1\leq i,j\leq k}=\prod_{1\leq i<j\leq k}(x_{j}-x_{i}).
\end{equation}
A function $f$ is said to be a \textit{P\'olya frequency function of infinite order}, denoted by $f\in\mathrm{PF}_{\infty}$, if and only if $f$ is a P\'olya frequency function of all orders so that one may write
    \begin{align}
      \mathrm{PF}_{\infty}=\bigcap\limits_{n\geq 1}\mathrm{PF}_{n}.  
    \end{align}
While a general description of $\mathrm{PF}_{n}$ for $n>3$ is not known, a few elementary and structural results are available. Firstly, for all integers $n\geq 1$ 
\begin{align}
    \mathrm{PF}_{\infty}\subset\mathrm{PF}_{n+1}\subset\mathrm{PF}_{n}.
\end{align}
In particular, the first level of the hierarchy is simply given by
\begin{align}
    \mathrm{PF}_{1}=\left \{ f\in\mathrm{L}^{1}\left ( \R \right ):\forall x\in\R,\, f\left ( x \right )\geq 0 \right \}.
\end{align}
The second level $\mathrm{PF}_{2}$ also admits a complete characterization: a continuous and integrable function $f$ on $\R$ is in $\mathrm{PF}_{2}$ if and only if $f$ is positive and log-concave. A proof may be found in Ref.~\cite[Section~11]{saumardLogconcavityStrongLogconcavity2014}. We also signal that a complete, albeit more involved, characterization of $\mathrm{PF}_{3}$ is available in \cite{WeinbergerCharacterizationPolyaFrequency1983}.

Sometimes coined \textit{Polynomial Ensembles of Derivative Type} in the literature, \textit{P\'olya ensembles} on $\mathrm{GL}_{N}\left(\C\right)$ are families of matrix-valued random variables parametrized by a P\'olya frequency function $\omega$ which are included in polynomial ensembles firstly introduced in~\cite{kuijlaarsTransformationsPolynomialEnsembles2016}, themselves included in the larger family of bi-unitary invariant (sometimes called \textit{isotropic}) random matrices \cite{kieburgProductsRandomMatrices2019,forsterPolynomialEnsemblesPolya2021}. 
The bi-unitary invariance means that their probability measure on $\mathrm{GL}_{N}\left(\C\right)$ is invariant under the adjoint action of the unitary group $\mathrm{U}(N)$ in the following sense
\begin{align}\label{def:biUinvariance}
    \forall M\in\mathrm{GL}_{N}\left(\C\right),\,\forall\, U_{1},U_{2} \in\mathrm{U}(N):\quad \mathrm{d}\mathbb{P}\left ( U_{1}MU_{2} \right )=\mathrm{d}\mathbb{P}\left ( M \right )
\end{align}
This property implies that the associated joint probability density function (\textbf{jpdf}) of the  complex eigenvalues denoted by $f_{\mathrm{EV}}$ is rotationally-invariant. It also implies that the singular values are independent of their singular vectors, each of the latter being Haar distributed on the unit sphere of $\C^{N}$. It happens that the distribution of singular values looks simpler when one squares them, the corresponding squared singular values jpdf will be denoted by  $f_{\mathrm{SV}}$. 

We say that a function $\omega\in\mathrm{L}^{1,N-1}_{\left [ 1,N \right ]}\left ( \R_{+} \right )$ gives rise to a \textit{P\'olya Ensemble} on ${\rm Gl}_{N}\left(\C\right)$ as in  \cite[Eq.~(2.17)]{forsterPolynomialEnsemblesPolya2021} (equivalently that $\omega$ is a \textit{P\'olya weight}), if and only if the logarithmic deformation
\begin{align}\label{def:polyaweight}
    \widetilde{\omega}:x\mapsto e^{-x}\omega\left ( e^{-x} \right )
\end{align}
is in $\mathrm{PF}_{N}$ and the squared singular values jpdf is
\begin{equation}\label{def:SVjpdf}
    f^{\left ( N \right )}_{\mathrm{SV}}\left [ \omega \right ](\mathbf{x})= C^{\left ( N \right )}_{\mathrm{SV}}\left [ \omega \right ]\Delta_{N}\left ( \mathbf{x} \right )\det\left(\Theta^{\left(j-1\right)}\left [ \omega \right ]\left ( x_{i} \right )\right)_{1\leq i,j\leq N}.
\end{equation}
The normalisation constants corresponds to the first term of 
\begin{align}\label{normalization1}
    C^{\left ( N \right )}_{\mathrm{SV}}\left [ \omega \right ]=\Big(\prod\limits_{k=1}^{N}k!\mathcal{M}\left [ \omega \right ]\left ( k \right )\Big)^{-1}, \quad C^{\left ( N \right )}_{\mathrm{EV}}\left [ \omega \right ]=\Big(\pi^N N!\prod\limits_{k=1}^{N}\mathcal{M}\left [ \omega \right ]\left ( k \right )\Big)^{-1}.
\end{align}
The second term normalizes the complex eigenvalues jpdf given by
\begin{equation}\label{def:EVjpdf}
            f^{\left ( N \right )}_{\mathrm{EV}}\left [ \omega \right ](\mathbf{z})= C^{\left ( N \right )}_{\mathrm{EV}}\left [ \omega \right ]\left |\Delta_{N}\left ( \mathbf{z} \right )  \right |^{2}\prod\limits_{k=1}^{N}\omega\left ( \left | z_{k} \right |^{2} \right ).
\end{equation}
\begin{remark}
Without loss of generality, we normalise the P\'olya weights so that $\mathcal{M}[\omega](1)=1$.
\end{remark}
We say that a random matrix $X$ is drawn from the P\'olya ensemble (of multiplicative type on $\mathrm{GL}_{N}(\C)$ with weight $\omega$) if its probability density on $\mathrm{GL}_{N}(\C)$ is bi-unitarily invariant in the sense of \eqref{def:biUinvariance} and if the  jpdfs of the squared singular values as well as complex eigenvalues are given respectively by \eqref{def:SVjpdf} and \eqref{def:EVjpdf}. In that case, we write $X\sim\mathrm{P\acute{o}l}_{N}\left [ \omega \right ]$.

It should be noted that a P\'olya weight $\omega$, such that $\widetilde{\omega}$ is at least in $\rm{PF}_{2}$, can be represented with a logarithmic change of variables so that, for all $t>0$
\begin{align}\label{id:polyarep}
    \omega(t)=\frac{\exp(-\psi(-\rm{ln}(t)))}{t},
\end{align}
where $\psi$ is a continuous convex function on $\R$. Additionally, logarithmic deformations of P\'olya weights in $\mathrm{PF}_{2}$ decay at least exponentially near infinity, see for example Lemma ~A.1 in \cite{DumbgenMaximumLikelihoodEstimation2009}.

Before giving an example, we conclude this section with a crucial result that sheds light on the multiplicative structure of the class of random matrices formed by \textit{P\'olya ensembles}. For this purpose, we choose two P\'olya weights $\omega_{1}$ and $\omega_{2}$ such that the corresponding functions $\widetilde{\omega}_{1}$ and $\widetilde{\omega}_{2}$ defined in~\eqref{def:polyaweight} belong to $\mathrm{PF}_{\infty}$. Moreover, we draw two independent P\'olya random $N\times N$ matrices $X\sim\mathrm{P\acute{o}l}_{N}\left [ \omega_{1} \right ]$ and $Y\sim\mathrm{P\acute{o}l}_{N}\left [ \omega_{2} \right ]$ and set
        \begin{align}
            \check{\omega}_{1}(x)=\frac{1}{\mathcal{M}[\omega_{1}](N)x^{N+1}}\omega_{1}\left ( \frac{1}{x} \right ).
        \end{align}     
        
\begin{proposition}[see Eq.~(3.10), Rem.~3.5 in~\cite{kieburgProductsRandomMatrices2019}]\label{polya1}
        We have
        \begin{align}
            XY\sim\mathrm{P\acute{o}l}_{N}\left [ \omega_{1}\circledast \omega_{2} \right ] \qquad{\rm and}\qquad X^{-1}\sim\mathrm{P\acute{o}l}_{N}\left [ \check{\omega}_{1} \right ].
        \end{align}
        In particular, for the generalized ratio of $X$ and $Y$ we have $X^{-1}Y\sim\mathrm{P\acute{o}l}_{N}\left [ \check{\omega}_{1}\circledast \omega_{2} \right ]$ and for every complex number $z$ where both $\mathcal{M}\left [ \omega_{1} \right ](z)$ and $\mathcal{M}\left [ \omega_{2} \right ](N+1-z)$ exist the following holds
        \begin{align}
            \mathcal{M}\left [ \check{\omega}_{1}\circledast \omega_{2} \right ]\left ( z \right )=\frac{\mathcal{M}\left [ \omega_{1} \right ]\left ( N+1-z \right )\mathcal{M}\left [ \omega_{2} \right ]\left ( z \right )}{\mathcal{M}[\omega_1](N)}.
        \end{align}
\end{proposition}
The previous proposition is crucial in the ensuing analysis and allows us to explicitly reduce the 2-matrix model to a 1-matrix model. For a given P\'olya weight $\omega$ we define the self-inverse-convoluted P\'olya weight by
\begin{equation}\label{omegahat}
    \widehat{\omega}(z)= \check{\omega}\circledast \omega\left(z\right).
\end{equation}
The latter may be written explicitly as follows
\begin{equation}\label{combi}
\widehat\omega(x)=\int_0^\infty \omega(y)\left(\frac{y}{x}\right)^{N+1}\frac{\omega\left(y/x\right)}{\mathcal{M}[\omega](N)}  \frac{\mathrm{d}y}{y}\overset{y\to xy}{=}\int_0^\infty \frac{\omega(xy)\omega(y)}{\mathcal{M}[\omega](N)} y^{N} \mathrm{d}y,
\end{equation}
it should be noted that $\widehat{\omega}$ implicitly depends on $N$.

A self-inverse-convoluted P\'olya weight has the following symmetry
\begin{equation}\label{symmetry}
\widehat\omega(x^{-1})=\int_0^\infty  \frac{\omega(y/x)\omega(y)}{\mathcal{M}[\omega](N)} y^{N}\mathrm{d}y\underset{y=tx}{=}x^{N+1}\int_0^\infty \frac{\omega(t)\omega(xt)}{\mathcal{M}[\omega](N)} t^{N} \mathrm{d}t=x^{N+1}\widehat\omega(x).
\end{equation}
Consequently, the Mellin transform of a self-inverse-convoluted P\'olya weight, together with its incomplete analogue, satisfies the following reflection symmetries:
    \begin{align}
    \mathcal{M}\left[\widehat{\omega}\right]\left(z\right)=&\frac{\mathcal{M}\left [ \omega \right ]\left ( N-z+1 \right )\mathcal{M}\left [ \omega \right ]\left ( z \right )}{\mathcal{M}[\omega](N)}=\mathcal{M}\left[\widehat{\omega}\right]\left(N-z+1\right),\label{Mellin.reflect}
        \\
        \mathcal{M}\left[\widehat{\omega}\right]\left(z,A\right)=&\mathcal{M}\left[\widehat{\omega}\right]\left(z\right)-\int_{A}^\infty \widehat{\omega}(x) x^{z-1}\mathrm{d}x\nonumber\\
        \overset{x\to x^{-1}}{=}&\mathcal{M}\left[\widehat{\omega}\right]\left(z\right)-\mathcal{M}\left[\widehat{\omega}\right]\left(N-z+1,A^{-1}\right).\label{Mellin.icomplete.reflect}
    \end{align}

\begin{example}\label{ex:muttalib-borodin}
A central non-Gaussian example is the Muttalib-Borodin ensemble of Laguerre type~\cite{borodinBiorthogonalEnsembles1998} closely connected to quantum transport models via the Dorokhov-Mello-Pereyra-Kumar (\textbf{DMPK}) equation~\cite{dorokhovTransmissionCoefficientLocalization1982, melloMacroscopicApproachMultichannel1988}. The corresponding choice of the P\'olya weight, after re-labelling to avoid confusion with our own notation, is
\begin{align}\label{def:MB-weight}
    \omega_{\mathrm{MB}}\left(t\right)=\frac{\alpha\gamma}{\Gamma((\delta+1)/\gamma)}(\alpha t)^{\delta}e^{-(\alpha t)^{\gamma}},\,\text{where }\,\alpha>0,\,\gamma>0,\ \delta>-1
\end{align}
Special choices of parameters recover familiar ensembles: setting $\alpha=1$, $\delta=0$ and $\gamma=1$ yields the GinUE while $\alpha=1$, $\delta\neq 0$ and $\gamma=1$ corresponds to the induced GinUE (see for example ~\cite{fischmannInducedGinibreEnsemble2012}).
The Muttalib-Borodin ensemble of Laguerre type has been studied extensively in \cite[Sec.~4]{borodinBiorthogonalEnsembles1998} and the corresponding  jpdf of its complex eigenvalues is
\begin{align}\label{muttalibborodinev}
   f^{\left ( N \right )}_{\mathrm{EV}}\left [ \omega_{\mathrm{MB}} \right ]\left(z_{1},\ldots,z_{N}\right)\propto |\Delta_{N}\left( z_{1},\ldots,z_{N} \right)|^2\prod\limits_{j=1}^{N}|z_{j}|^{2\delta}e^{-\alpha |z_{j}|^{2\gamma}},
\end{align}
where the proportionality constant can be computed, after multiplication by $\alpha^{N(\delta+1)}$, using~\eqref{normalization1} and the Barnes G-function.

The self-inverse-convoluted P\'olya weight defined at \eqref{omegahat} associated with the generalized ratio $X^{-1}Y$ where $X$ and $Y$ are two random matrices of size $N$ drawn from the Muttalib-Borodin ensemble of Laguerre type, whose P\'olya weight $\omega_{\rm{MB}}$ was previously defined in~\eqref{def:MB-weight}, is given by
\begin{align}\label{id:generalizedratioweight}
    \widehat{\omega}_{\mathrm{MB}}\left(t\right)=\frac{\gamma\Gamma\left( (N + 2\delta + 1)/\gamma \right)}{\Gamma((\delta+N)/\gamma)\Gamma((\delta+1)/\gamma)} \, t^{\delta} \left(1 + t^\gamma \right)^{ -(N+2\delta+1)/\gamma}
\end{align}
and the corresponding Mellin transform is given by
\begin{align}\label{Mellin-MBI}
    \mathcal{M}[\widehat{\omega}_{\mathrm{MB}}](z)=\frac{\Gamma((\delta+z)/\gamma)\Gamma((N+\delta+1-z)/\gamma)}{\Gamma((\delta+1)/\gamma)\Gamma((N+\delta)/\gamma)}.
\end{align}
We note that the width of the original distribution determined by $\alpha$ completely drops out which is not surprising since the corresponding product random matrix $X^{-1}Y$ is scale-invariant when $X$ and $Y$ are drawn from the same ensemble. This observation will certainly hold true for general P\'olya ensembles beyond Muttalib-Borodin ensembles.
\end{example}

We conclude this section by considering the $2$-matrix model~\eqref{def:2matrixmodel} with $K_1$ and $K_2$ independently drawn from the same P\'olya ensemble associated to the weight $\omega$. In order to safely consider its generalized partition function, we present a short lemma to ensure the existence thereof.

\begin{lemma}\label{lemma:existence}
    Let $\omega$ be a P\'olya weight such that $\widetilde{\omega}$ is a P\'olya frequency function of order $2$, for any $j\in[0,2N-2]$ 
    \begin{equation}
        \int_{\C}|z|^{j}\omega(|z|^2)\mathrm{d}^2z<\infty.
    \end{equation}
    Let $m$ be a strictly positive integer, $(\alpha_{\ell})$ be $m$ complex numbers and $(\beta_{\ell})$ be $m$ pairwise distinct  non-zero complex numbers.
    Then, the following holds
    \begin{equation}
        \int_{\C^{N}}\left|\Delta_{N}(\mathbf{z})\right|^{2}\prod\limits_{j=1}^{N}\left(\omega(\left|z_{j}\right|^{2})\prod\limits_{\ell=1}^{m}\left|\frac{z_{j}-\alpha_{\ell}}{z_{j}-\beta_{\ell}}\right|\right)\mathrm{d}^{2}\mathbf{z}<+\infty.
    \end{equation}
\end{lemma}

\begin{proof}
    Making use of the triangular inequality, we have the following rough bound
    \begin{equation}
        \left|\Delta_{N}(\mathbf{z})\right|^{2}\leq \prod\limits_{j=1}^{N}(1+\left|z_{j}\right|)^{2(N-1)}.
    \end{equation}
    Let us fix $r>0$ small enough so that all closed discs $\overline{\mathrm{D}}(\beta_{j},r)\subset\C$ are pairwise disjoint and do not intersect with the origin. We then set
    \begin{equation}
        \Omega:=\left(\C\setminus \bigcup\limits_{j=1}^{m}\overline{\mathrm{D}}(\beta_{j},r)\right)^{N}.
    \end{equation}
    For $\mathbf{z}\in\Omega$, each coordinate satisfies $\left|z_{j}-\beta_{\ell}\right|\geq r$ for all $1\leq \ell\leq m$. Therefore, uniformly in $z_{j}$ where $j$ is fixed, we have
    \begin{equation}
        \prod\limits_{\ell=1}^{m}\left|\frac{z_{j}-\alpha_{\ell}}{z_{j}-\beta_{\ell}}\right|=\prod\limits_{\ell=1}^{m}\left|1+\frac{\beta_{\ell}-\alpha_{\ell}}{z_{j}-\beta_{\ell}}\right|\leq \prod\limits_{\ell=1}^{m}\left(1+\frac{|\beta_{\ell}-\alpha_{\ell}|}{r}\right).
    \end{equation}
    By assumption $\int_{\C}|z|^{j}\omega(|z|^2)\mathrm{d}^2z<\infty$ for any $j\in[0,2N-2]$, the integral 
    \begin{equation}
    \begin{split}
        &\int_{\Omega}\left|\Delta_{N}(\mathbf{z})\right|^{2}\prod\limits_{j=1}^{N}\left(\omega(\left|z_{j}\right|^{2})\prod\limits_{\ell=1}^{m}\left|\frac{z_{j}-\alpha_{\ell}}{z_{j}-\beta_{\ell}}\right|\right)\mathrm{d}^{2}\mathbf{z}
        \\
        \leq&\left(\int_{\C\setminus \bigcup\limits_{j=1}^{m}\overline{\mathrm{D}}(\beta_{j},r)}(1+\left|z\right|)^{2(N-1)}\prod\limits_{\ell=1}^{m}\left(1+\frac{|\beta_{\ell}-\alpha_{\ell}|}{r}\right)\omega(|z|^2)\mathrm{d}^2z\right)^N<\infty
    \end{split}
    \end{equation}
    exists on $\Omega$.

    On the complement $\Omega^{\complement}$, there exists an index pair $(j,\ell)$ such that $\left|z_{j}-\beta_{\ell}\right|< r$. We work inside the compact closed disc 
    \begin{equation}
        \overline{\mathrm{D}}(\beta_{\ell},r)=\left\{ z\in\mathbb{C}:\,\left| z-\beta_{\ell} \right|\leq r \right\}.
    \end{equation}
    Because $\omega$ is a P\'olya weight, it is positive, continuous and integrable on $\mathbb{R}_{+}$. In particular, any singularity must arise at a boundary point say $t_{0}\in\mathbb{R}_{+}$: there exists $\alpha\in[0,1[$ and $c_{1}>0$ such that, on a neighbourhood of $t_{0}$
    \begin{equation*}
        \omega(t)\leq c_{1}\left|t-t_{0}\right|^{-\alpha}
    \end{equation*}
    If $\left|\beta_{\ell}\right|^{2}\neq t_{0}$, we simply have $\alpha=0$ in what follows. Assume there exist an index $\ell$ such that $\left|\beta_{\ell}\right|^{2}= t_{0}$. Because $\beta_{\ell}$ is non-zero and the function $z\mapsto \left|z\right|^{2}$ has a non-vanishing gradient at $z=\beta_{\ell}$, for $r>0$ sufficiently small there exist constants $c_{2},c_{3}>0$ such that
    \begin{equation*}
        c_{2}\left|z-\beta_{\ell}\right|\leq \left|\left|z\right|^{2}-\left|\beta_{\ell}\right|^{2}\right| \leq c_{3}\left|z-\beta_{\ell}\right|.
    \end{equation*}
    Consequently, on $\overline{\mathrm{D}}(\beta_{\ell},r)$ there exists $c_{4}>0$ such that
    \begin{equation*}
        \omega(\left|z\right|^{2})\leq c_{4}\left|z-\beta_{\ell}\right|^{-\alpha}.
    \end{equation*}
    All factors in the integrand that remain away from their own singularities are bounded on this disc. In particular,
    \begin{itemize}
        \item all denominators $\left|z_{j}-\beta_{k}\right|$ with $k\neq \ell$ stay uniformly bounded away from zero since the discs are chosen disjoints,
        \item the Vandermonde term is bounded by
        \begin{equation*}
            \left|\Delta_{N}(\mathbf{z})\right|^{2}\leq \prod_{j=1}^{N}(1+\left|z_{j}\right|)^{2(N-1)}.
        \end{equation*}
    \end{itemize}
    It follows that, for $z_{j}\in \overline{\mathrm{D}}(\beta_{\ell},r)$, the factor
    \begin{equation*}
        (1+\left|z_{j}\right|)^{2(N-1)}\omega(\left|z_{j}\right|^{2})\left|z_{j}-\alpha_{\ell}\right|\prod\limits_{k\neq\ell}^{}\left|\frac{z_{j}-\alpha_{k}}{z_{j}-\beta_{k}}\right|
    \end{equation*}
    is bounded, up to a constant, by $\left|z_{j}-\beta_{\ell}\right|^{-\alpha}$. Hence, there exists $c_{5}>0$ such that the full integrand on $\overline{\mathrm{D}}(\beta_{\ell},r)$ is dominated by
    \begin{equation*}
        c_{5}\left|z_{j}-\beta_{\ell}\right|^{-(1+\alpha)}.
    \end{equation*}
    In polar coordinates around $\beta_{\ell}$, that is $z=\beta_{\ell}+\rho e^{i\theta}$, on this closed disc we have
    \begin{equation*}
        \int_{\overline{\mathrm{D}}(\beta_{\ell},r)}\left|z_{j}-\beta_{\ell}\right|^{-(1+\alpha)}\mathrm{d}^{2}z=\int_{0}^{2\pi}\mathrm{d}\theta\int_{0}^{r}\rho^{-(1+\alpha)}\rho\mathrm{d}\rho=2\pi\int_{0}^{r}\rho^{-\alpha}\mathrm{d}\rho <\infty,
    \end{equation*}
    since $\alpha<1$. Therefore, the integral over $\Omega^{\complement}$ is finite which allows us to conclude.
    
\end{proof}

We now define the generalized version of the partition function of our model \eqref{partitionfunction1}, which has been extended to allow the capture of possible real and imaginary parts of the winding number density in the correlation functions
\begin{align}\label{partitionfunction2}
    \mathscr{Z}^{\left(N\right)}_{m_{1},m_{2}}\left(\textbf{p},\textbf{q},\widetilde{\textbf{p}},\widetilde{\textbf{q}}\right)=\mathbb{E}\Bigg(\prod\limits_{j=1}^{m_{1}}\frac{\det\left ( K\left ( p_{j} \right ) \right )}{\det\left ( K\left ( q_{j} \right ) \right )}\prod\limits_{\ell=1}^{m_{2}}\frac{\overline{\det\left ( K\left ( \widetilde{p}_{\ell} \right ) \right )}}{\overline{\det\left ( K\left ( \widetilde{q}_{\ell} \right ) \right )}}\Bigg),
\end{align}
where $\mathbf{p},\mathbf{q},\widetilde{\mathbf{p}}$ and $\widetilde{\mathbf{q}}$ are four collections of points on the unit circle $\Sp$ such that 
\begin{align}\label{id:pointsseparation}
    \text{for }\,1\leq j<\ell\leq m_{1}: a(q_j)b(q_{\ell})\neq a(q_{\ell})b(q_j)\nonumber,
    \\
    \text{for }\,1\leq j<\ell\leq m_{2}: a(\widetilde{q}_j)b(\widetilde{q}_{\ell})\neq a(\widetilde{q}_{\ell})b(\widetilde{q}_j),
    \\
    \text{for }\,1\leq j\leq m_{1},\,1\leq \ell\leq m_{2}: a(q_j)b(\widetilde{q}_{\ell})\neq a(\widetilde{q}_{\ell})b(q_j)\nonumber,
\end{align}
and we recall that $K$ as in \eqref{def:2matrixmodel} is an additive 2-matrix model where the two source matrices $K_{1}$ and $K_{2}$ are independently drawn from two P\'olya ensembles with respective weights $\omega_{1}$ and $\omega_{2}$. Observing that for any pair of points on $\Sp$, say $p$ and $q$, one has
\begin{equation}
    \mathbb{E}\Bigg(\frac{\det(K(p))}{\det(K(q))}\Bigg)=\Bigg(\frac{b(p)}{b(q)}\Bigg)^{N}\mathbb{E}\Bigg(\frac{\det(\kappa(p)I_{N}+K_{1}^{-1}K_{2})}{\det(\kappa(q)I_{N}+K_{1}^{-1}K_{2})}\Bigg),
\end{equation}
where $\kappa$ has been introduced at \eqref{functions1}. By construction, the distribution of the random matrix $K_{1}^{-1}K_{2}$ is that of a P\'olya ensemble, this time with a P\'olya weight given by $\check{\omega}_{1}\circledast\omega_{2}$ so we may rewrite the previous generalized partition function as
\begin{equation}\label{partitionfunction3}
\begin{split}
    &\mathscr{Z}^{\left(N\right)}_{m_{1},m_{2}}\left(\textbf{p},\textbf{q},\widetilde{\textbf{p}},\widetilde{\textbf{q}}\right)=\Bigg( \prod\limits_{j=1}^{m_{1}}\frac{b\left ( p_{j} \right )}{b\left ( q_{j} \right )}\prod\limits_{\ell=1}^{m_{2}}\frac{\overline{b\left ( \widetilde{p}_{\ell} \right )}}{\overline{b\left ( \widetilde{q}_{\ell} \right )}} \Bigg)^{N} 
    \\
    \times&\mathbb{E}\Bigg(\prod_{j=1}^{m_{1}}\frac{\det\left ( \kappa\left ( p_{j} \right )I_{N}+K_{1}^{-1}K_{2} \right )}{\det\left ( \kappa\left ( q_{j} \right )I_{N}+K_{1}^{-1}K_{2} \right )}\prod_{\ell=1}^{m_{2}}\frac{\det\left ( \overline{\kappa\left ( \widetilde{p}_{\ell} \right )}I_{N}+\overline{K_{1}^{-1}K_{2}  } \right )}{\det\left ( \overline{\kappa\left ( \widetilde{q}_{\ell} \right )}I_{N}+\overline{K_{1}^{-1}K_{2}} \right )}\Bigg).
\end{split}
\end{equation}
Making use of the change-of-variables formula for expectations, one may then apply \autoref{lemma:existence} with 
\begin{align*}
    \alpha_{j}=-\kappa(p_{j})\,\text{ if }\,1\leq j\leq m_{1},\,\text{ and }\,\alpha_{j}=-\overline{\kappa(\widetilde{p}_{j-m_{1}})}\,\text{ if }\,m_{1}+1\leq j\leq m_{1}+m_{2},
    \\
    \beta_{j}=-\kappa(q_{j})\,\text{ if }\,1\leq j\leq m_{1},\,\text{ and }\,\beta_{j}=-\overline{\kappa(\widetilde{q}_{j-m_{1}})}\,\text{ if }\,m_{1}+1\leq j\leq m_{1}+m_{2},
\end{align*}
the set of conditions \eqref{id:pointsseparation} ensuring that the $\beta_{j}$ are pairwise distinct and $\check{\omega}_{1}\circledast\omega_{2}$ playing the role of $\omega$ to safely ensure the existence of the generalized partition function \eqref{partitionfunction2}. The case of vanishing $a(p)$ or $b(p)$ will be understood in terms of analytic continuation to these points if the P\'olya weight allows this. Indeed, we consider a class for which this continuation is possible.

One naturally recovers~\eqref{partitionfunction1} by setting $m_{2}=0$ with the convention that the product over an empty index set is $\prod_{\emptyset}X=1$.

\subsection{Main results}\label{sec:main-results}

To state our first main result, we define the Cauchy-like kernel
\begin{align}\label{cauchykernel}
    \mathrm{Q}_{m}(\mathbf{p}, \mathbf{q})=\Bigg(\frac{1}{\nu\left ( p_{i} \right )^{\top}\mathrm{J}\nu\left ( q_{j} \right )}\Bigg)_{1\leq i,j\leq m}
\end{align}
whose determinant already appeared in the denominator of \eqref{partitiongaussian} and that satisfies the anti-symmetry relation $\mathrm{Q}_{m}(\mathbf{p}, \mathbf{q})^{\top}=-\mathrm{Q}_{m}(\mathbf{q}, \mathbf{p})$.
We also define the following auxiliary function
\begin{align}\label{ypsilon.def}
    \Upsilon_{N}\left(u,v\right)=\sum\limits_{k=1}^{N}\frac{\mathcal{M}\left [ \widehat{\omega} \right ]\left ( k,\left |v \right |^{2} \right )}{\mathcal{M}\left [ \widehat{\omega} \right ]\left ( k \right )}\Big(\frac{u}{v}\Big)^{k}
\end{align}
which in turn we use to define the following deformed kernel
\begin{align}\label{polyakernel}
    \widetilde{\mathrm{Q}}^{(N)}_{m}[\omega](\mathbf{p}, \mathbf{q}) = \left( 
\frac{\left( b(p_{i})/b(q_{j}) \right)^N}{\nu(p_{i})^\top \mathrm{J} \nu(q_{j})} 
\left[ 1 + \Big(1-\frac{\kappa\left( q_{j} \right)}{\kappa\left( p_{i} \right)}\Big) 
\Upsilon_N\big( \kappa(p_{i}), \kappa(q_{j}) \big) 
\right]
\right)_{1 \leq i,j \leq m},
\end{align}
where $\nu,\,\kappa$ and $J$ are as defined in \eqref{functions1}.
\begin{remark}
    The auxiliary function $\Upsilon_{N}$ satisfies the following symmetry relations
    \begin{align}\label{ypsilon.sym.1}
        \Upsilon_{N}\left( u,v \right)=& \frac{u}{v} \Bigg(\frac{\left( u/v \right)^{N}-1}{u/v-1}-\left( \frac{u}{v} \right)^{N}\Upsilon_{N}\left( u^{-1},v^{-1} \right)\Bigg) &&\text{if $u\neq v$},
        \\
        \Upsilon_{N}\left( u,u \right)=&N-\Upsilon_{N}\left( u^{-1},u^{-1} \right),\label{ypsilon.sym.2}
    \end{align}
which are direct consequences of~\eqref{Mellin.reflect} and~\eqref{Mellin.icomplete.reflect} by reflecting $k\leftrightarrow N-k+1$ in the summation index.
\end{remark}
\begin{remark}
    Expression \eqref{polyakernel} is not manifestly symmetric under the exchange of $a$ and $b$, nevertheless the underlying 2-matrix model \eqref{def:2matrixmodel} is symmetric and this invariance is preserved in a non-obvious way. More precisely, combining \eqref{Mellin.reflect} and \eqref{ypsilon.sym.1} allows to rewrite the deformed kernel $\widetilde{\mathrm{Q}}^{(N)}_{m}$ as a sum of eight terms for which the functions $a$ and $b$ clearly play symmetrical roles. This expression is, however, too complicated to manipulate in further computations, so we deliberately retain the simpler asymmetric form which proves more natural and effective for the subsequent asymptotic analysis.
\end{remark}
Our first main result is the following theorem.
\begin{theorem}\label{theorem1}
The partition function \eqref{partitionfunction1} associated with an additive 2-matrix model where the two source matrices $K_{1}$ and $K_{2}$ are independently drawn from a Pólya ensemble of size $N$ with weight $\omega$ is given by
    \begin{align}
            \mathscr{Z}^{\left(N\right)}_{m}\left ( \mathbf{p},\mathbf{q} \right )=\frac{\det\left ( \widetilde{\mathrm{Q}}^{(N)}_{m}[\omega](\mathbf{p}, \mathbf{q})\right)}{\det\left ( \mathrm{Q}_{m}\left ( \mathbf{p},\mathbf{q} \right ) \right )}.
        \end{align}
\end{theorem}

In fact, we will prove a more general form of this result (\autoref{theorem3}) in \autoref{sec1}. We nevertheless state this simpler version here for brevity. From that result, we next derive the asymptotic expansion for the first moment of the winding number~\eqref{windingnumber1}.
 
In the Gaussian case, the following exact expression holds for the 2-matrix model~\eqref{def:2matrixmodel} corresponding to the partition function~\eqref{partitiongaussian}
\begin{align}\label{asymptoticwinding1}
    \mathbb{E}\left(\mathrm{Wind}_{N} \right)=\Bigg(\frac{1}{2\pi i}\oint_{\mathrm{S}^{1}}^{}\frac{\nu\left( p \right)^{\dagger}\nu'\left( p \right)}{\left\| \nu\left( p \right) \right\|^{2}}\mathrm{d}p\Bigg)N.
\end{align}
 
The factor in parentheses corresponds to the \textit{Aharonov-Anandan phase angle} associated with the complex line bundle
\begin{equation}
    \pi:\mathrm{span}_{\C}(\nu)\to \Sp,
\end{equation}
whose fiber at $p$ is the complex line generated by $\nu(p)$.

To make this geometric interpretation explicit, write
\begin{equation}
    r(p):=\left\| \nu(p) \right\|,\quad \psi(p):=\frac{\nu(p)}{\left\| \nu(p) \right\|},
\end{equation}
so that
\begin{equation}
    \frac{\nu^{\dagger}\nu'}{\nu^{\dagger}\nu}=\frac{(r\psi)^{\dagger}(r'\psi+r\psi')}{r^{2}}=\frac{r'}{r}+\psi^{\dagger}\psi'.
\end{equation}
The first term on the right hand side integrates to zero since $r$ is a smooth, positive function on $\Sp$, i.e.,
\begin{equation}
    \oint_{\Sp}\frac{r'(p)}{r(p)}\mathrm{d}p=0.
\end{equation}
Differentiating the identity $\psi^{\dagger}\psi=1$ with respect to $p$ shows that $\psi^{\dagger}\psi'$ is purely imaginary. Hence, there exists a real-valued function $\mathcal{A}$ such that $\psi^{\dagger}\psi'=i\mathcal{A}$ which defines the \textit{Aharonov-Anandan connection}
\begin{equation}
    \mathcal{A}:=i\left<\psi(p),\frac{\mathrm{d}}{\mathrm{d}p}\psi(p)\right>=i\psi(p)^{\dagger}\psi'(p),
\end{equation}
and consequently
\begin{equation}
    \frac{1}{2\pi i}\oint_{\mathrm{S}^{1}}^{}\frac{\nu\left( p \right)^{\dagger}\nu'\left( p \right)}{\left\| \nu\left( p \right) \right\|^{2}}\mathrm{d}p=\frac{1}{2\pi}\oint_{\Sp}\mathcal{A}(p)\mathrm{d}p.
\end{equation}

This expression sheds light on the geometric nature of the coefficient in \eqref{asymptoticwinding1}: once multiplied by $2\pi$ and exponentiated it yields the holonomy (i.e., total phase change) acquired by a normalised section of the line bundle $\pi$ transported once along the unit circle.
In the adiabatic limit, this phase reduces to the \textit{Berry phase}. In the present setting, however, we do not restrict ourselves to adiabatic evolution so it should be interpreted as an Aharonov-Anandan phase.

A thorough treatment of these notions: Aharonov-Anandan and Berry phases, holonomy, vector bundles and their realisations in quantum systems may be found in \cite[Chap.~4-7,12]{BohmGeometricPhaseQuantum2003}. 
 
Moving beyond the Gaussian case, we will choose the P\'olya weight \eqref{def:polyaweight} to be
\begin{align}
    \omega_{\mathrm{MB}}\left(t\right)=\frac{\gamma}{\Gamma((\delta+1)/\gamma)}t^{\delta}e^{-t^{\gamma}},
\end{align}
to get an explicit expression. As aforementioned, it corresponds to the Muttalib-Borodin Ensemble of Laguerre type, see~\eqref{ex:muttalib-borodin}. It should be noted that $\omega$ can be rewritten as \eqref{id:polyarep} with
\begin{align*}
    \psi(x)=e^{-\gamma x}+(\delta-1)x,\,\delta>-1\text{ and }\gamma>0, 
\end{align*}
which is a strictly convex function on $\R$ so that $\omega_{\rm{MB}}$ is indeed in $\rm{PF}_{2}$ and thus a P\'olya weight.

When we introduce the modified vector field
\begin{align}\label{id:nugamma}
    \nu_{\gamma}\left(p\right)=\begin{pmatrix}
        |a\left(p\right)|^{\gamma-1}a(p)
        \\
        |b\left(p\right)|^{\gamma-1}b(p)
    \end{pmatrix},
\end{align}
we find that for the Muttalib-Borodin case the winding number has a very similar form compared to the Gaussian case but has also its differences. Its asymptotic expansion as $N$ goes to infinity is summarised in our second main result.

\begin{theorem}\label{theorem2}
The average winding number associated with the determinantal curve of an additive $2$-matrix model where the two source matrices are independently drawn from a Muttalib-Borodin of Laguerre type ensemble admits, as $N$ goes to infinity, the following asymptotic expansion
    \begin{align}\label{windingasymptotic2}
\begin{split}
    \mathbb{E}\left(\mathrm{Wind}_{N} \right)=&\Bigg(\frac{1}{2\pi i}\oint_{\mathrm{S}^{1}}^{}\frac{\nu_{\gamma}\left( p \right)^{\dagger}\nu_{\gamma}'\left( p \right)}{\left\| \nu_{\gamma}\left( p \right) \right\|^{2}}\mathrm{d}p\Bigg)N\\
    &+\frac{\gamma-1-2\delta}{2}\Bigg(\frac{1}{2\pi i}\oint_{\mathrm{S}^{1}}^{}\frac{\kappa'\left( p \right)}{\kappa\left( p \right)}\frac{| a\left( p \right) |^{2\gamma}-| b\left( p \right) |^{2\gamma}}{| a\left( p \right) |^{2\gamma}+| b\left( p \right) |^{2\gamma}}\mathrm{d}p\Bigg)+o\left(1\right).
\end{split}    
\end{align}
where the coefficient appearing in the leading order term may be rewritten as
\begin{align}\label{windingasymptotic2.b}
    \frac{1}{2\pi i}\oint_{\mathrm{S}^{1}}^{}\frac{\nu_{\gamma}\left( p \right)^{\dagger}\nu_{\gamma}'\left( p \right)}{\left\| \nu_{\gamma}\left( p \right) \right\|^{2}}\mathrm{d}p=\frac{1}{2\pi i}\oint_{\mathrm{S}^{1}}^{}\frac{|a\left( p \right)|^{2\gamma-2} a'\left( p \right)\overline{a\left( p \right)}+|b\left( p \right)|^{2\gamma-2} b'\left( p \right)\overline{b\left( p \right)}}{|a\left( p \right)|^{2\gamma} +|b\left( p \right)|^{2\gamma} }\mathrm{d}p.
\end{align}
\end{theorem}

\begin{remark}
    The integral appearing in the subleading order correction term must be interpreted in the sense of a Cauchy principal value integral whenever $a$ or $b$ vanishes. Since we assume that both functions have only isolated zeroes on $\Sp$, their logarithmic derivatives possess at most simple poles. Consequently, the Cauchy principal value ensures that these singularities remain integrable.
\end{remark}

It should be noted that $\left\| \nu_{\gamma}\left( p \right) \right\|$ is non-zero for all $p\in\Sp$, since by assumption the parameter functions $a$ and $b$ never vanish simultaneously. The term \eqref{windingasymptotic2.b} remains an Aharonov-Anandan phase angle, now associated with the deformed vector field $\nu_{\gamma}$. The additional constant term is genuinely new and does not occur in the Gaussian case $\gamma=1$, cf.~\eqref{asymptoticwinding1}, indeed it vanishes in the joint limit $\gamma\to 1$ and $\delta\to 0$. Interestingly, the tail behaviour and the repulsion from the origin can compensate each other in the subleading order contribution when the parameters are tuned such that $\gamma-1=2\delta$.

More interestingly, when neither $a$ nor $b$ vanishes on $\Sp$, the subleading order term admits a geometric interpretation as a deformed winding number of the curve traced by $\kappa=a/b$ around the origin.
The deformation is controlled by a weight that measures the local asymmetry between $a$ and $b$. The contribution takes values in $[-1,1]$, vanishes wherever $|a(p)|=|b(p)|$ and approaches unity when $\left|\ln(\kappa(p))\right|$ is large, that is, when one component of the vector $\nu$ dominates in magnitude.

Finally, this interpretation as a weighted winding number around the origin ceases to apply when either $a$ or $b$ vanishes, since $\kappa$ takes values in the projective space $\mathbb{CP}^{1}\cong\mathrm{S}^{2}$ which is homotopically trivial. In that case, a winding number around a given point is no longer well-defined.

\begin{example}
   We conclude by illustrating the previous \autoref{theorem2} with a simple and explicit choice of parameter functions $a$ and $b$,
    \begin{align}
        a\left(p\right)=r_{1}p^n
        \qquad{\rm and}\qquad b\left(p\right)=r_{2}p^m,
    \end{align}
    where $r_{1},r_{2}>0$ and $n,m\in\mathbb{Z}$ are constants. A direct calculation shows that for $N\to\infty$ the following expansion holds
    \begin{align}
        \mathbb{E}\left(\mathrm{Wind}_{N} \right)=\Bigg(\frac{nr_{1}^{2\gamma}+mr_{2}^{2\gamma}}{r_{1}^{2\gamma}+r_{2}^{2\gamma}}\Bigg)N+\frac{\left( \gamma-1-2\delta \right)\left( n-m \right)}{2}\Bigg(\frac{r_{1}^{2\gamma}-r_{2}^{2\gamma}}{r_{1}^{2\gamma}+r_{2}^{2\gamma}}\Bigg)+o\left( 1 \right).
    \end{align}
    In the particular case where $r_{1}=r_{2}$, the average winding number becomes a half-integer which is noteworthy, although this behaviour does not necessarily hold in the general case.
\end{example}
\begin{remark}
    In the previous example it appears that the leading order may vanish while the subleading order term persists, for example setting $r_{1}=2^{1/2\gamma},\,r_{2}=1$ and $2n=-m$ yields an average winding number asymptotically equal to $n(\gamma-1-2\delta)/2$ as $N$ goes to infinity.
\end{remark}

\subsection{The Ginibre random field}\label{sec2.1}

We now analyse the partition function \eqref{partitionfunction1} in the case $m=1$ and its associated 1-point correlation function \eqref{C1.def} beyond the additive $2$-matrix model framework and consider a matrix-valued Gaussian random field over the unit circle $\Sp$, whose entries are independent and identically distributed complex Gaussian random fields. 

A complex Gaussian random field is uniquely determined by its mean, covariance and pseudo-covariance kernels. In the present case the mean function is identically zero and the covariance/pseudo-covariance kernels are given by
\begin{equation}\label{matrixgrf}
    \mathbb{E}\left(K_{i,j}\left( p \right)\overline{K_{k,\ell}}\left( q \right)  \right)=\Cc(p,q)\delta_{i,k}\delta_{j,\ell}\,\text{ and }\,\mathbb{E}\left( K_{i,j}\left( p \right)K_{k,\ell}\left( q \right) \right)=0.
\end{equation}
for all indices $i,j,k,\ell\in\left[ 1,N \right]$ and points $p,q\in\mathrm{S}^{1}$.
The covariance kernel $\Cc$ is Hermitian and positive-definite, in particular:
\begin{equation}
    \Cc(q,p)=\overline{\Cc(p,q)}\,\text{ and }\,\Cc(p,p)\geq 0.
\end{equation}
At first glance the choice of a null pseudo-covariance function might seem restrictive but it actually reflects the independence of the real and imaginary parts of each entry in the definition of the Ginibre Unitary Ensemble. This setting naturally generalizes the construction studied in \cite{hahnWindingNumberStatistics2023,hahnUniversalCorrelations2024}, which may be recovered by specifying
\begin{equation}
    \Cc(p,q)=\nu(q)^{\dagger}\nu(p).
\end{equation}

A random field $K$ as in \eqref{matrixgrf} may be referred to as a \textit{complex Ginibre random field} (of size $N$ with covariance kernel $\Cc$ over $\Sp$). To analyse functionals of this random field, we sample it at finitely many points on the base manifold $\Sp$, say $n\geq 1$. The result joint distribution of the random matrices $K(p_{1}),\ldots,K(p_{n})$ admits a finite-dimensional representation in terms of $n$ independent GinUE-drawn random matrices. This construction is analogous to a Karhunen-Lo\`eve expansion, see for example \cite[Sec.~3.2]{taylorRandomFieldsGeometry2007}.
\begin{lemma}\label{technicallemma1}
    Let $n$ be a strictly positive integer and $\mathbf{p}=\left(p_{1},\cdots,p_{n}\right)$ be a collection of $n$ points on the unit circle $\Sp$. We assume that the covariance function $\Cc$ is twice differentiable on $\mathbb{T}^{2}$.
    
    Then, there exist a unique lower-triangular matrix $\mathrm{L}$ of size $n$  whose entries are twice differentiable complex-valued functions on $\mathbb{T}^{n}$ such that for all indices $i,j$ in $[1,n]$
    \begin{equation}
        \Cc(p_{i},p_{j})=\sum_{k=1}^{\mathrm{min}(i,j)}\mathrm{L}_{i,k}(\mathbf{p})\overline{\mathrm{L}_{j,k}(\mathbf{p})}.
    \end{equation}
    If $G_{1},\ldots,G_{n}$ are $n$ random matrices identically and independently drawn from the ${\rm GinUE}$ of size $N$, then
    \begin{equation}\label{K.Gauss}
        \Big( K\left( p_{1} \right),\cdots,K\left( p_{n} \right) \Big)\overset{\rm d}{=}\left(\sum\limits_{j=1}^{n}\mathrm{L}_{1,j}(\mathbf{p})G_{j},\ldots,\sum\limits_{j=1}^{n}\mathrm{L}_{n,j}(\mathbf{p})G_{j}\right).
    \end{equation}
\end{lemma}
\begin{proof}
    Making use of the entrywise independence of the random field $K$, we start by fixing two indices $i,j$ in $[1,n]$ and defining the random vector
    \begin{equation*}
        K_{i,j}(\mathbf{p})=(K_{i,j}(p_{1}),\ldots,K_{i,j}(p_{n}))^{\top}\in\C^{n}.
    \end{equation*}
    Its covariance matrix is $\Cc_{n}(\mathbf{p})=(\Cc(p_{k},p_{\ell}))_{1\leq k,\ell\leq n}$ which is Hermitian and positive definite. Therefore the Cholesky decomposition ensures the existence of a unique lower-triangular matrix $\mathrm{L}(\mathbf{p})$ with positive diagonal entries such that $\Cc_{n}(\mathbf{p})=\mathrm{L}(\mathbf{p})\mathrm{L}(\mathbf{p})^{\dagger}$.

    Consequently, the random vector $\mathrm{L}(\mathbf{p})^{-1}K_{i,j}(\mathbf{p})$ is formed of independent complex standard Gaussian random variables, from which the stated distributional equality \eqref{K.Gauss} follows.

    The differentiability of $\mathrm{L}_{i,j}$ in $\mathbf{p}$ is inherited from the second-order differentiability of $\Cc$ through the recursive construction of the Cholesky coefficients. Explicitly, for $1\leq j<i\leq n$ we have
    \begin{equation}\label{L.constr}
    \begin{split}
        \mathrm{L}_{1,1}(\mathbf{p})=&\sqrt{\Cc(p_1,p_1)},\\
        \mathrm{L}_{i,i}(\mathbf{p})=&\sqrt{\Cc(p_{i},p_{i})-\sum_{k=1}^{i-1}|\mathrm{L}_{i,k}(\mathbf{p})|^2},\\
        \mathrm{L}_{i,j}(\mathbf{p})=&\mathrm{L}_{i,i}(\mathbf{p})^{-1}\left(\Cc(p_{i},p_{j})-\sum_{k=1}^{j-1}\overline{\mathrm{L}_{i,k}(\mathbf{p})}\mathrm{L}_{j,k}(\mathbf{p})\right),
    \end{split}
    \end{equation}
    coupled to the fact that $\Cc$ being positive-definite: all the principal minors of $\Cc_{n}(\mathbf{p})$ are positive. Therefore the square-roots as well as the inverses are well-defined.
\end{proof}

\begin{proposition}\label{proposition.gin}
Let $K$ be a centered complex Ginibre random field of size $N$ with covariance kernel $\Cc$ and null pseudo-covariance kernel over $\Sp$, and $p,q$ be two points on the unit circle. Then the partition function and the average winding number are given by
\begin{equation}
    \mathscr{Z}^{\left(N\right)}_{1}\left(p,q\right)=\Bigg(\frac{\Cc\left( p,q \right)}{\Cc\left( q,q \right)}\Bigg)^{N}\text{ and }\,\,\mathbb{E}\left( \mathrm{Wind}_{N} \right)=\Bigg(\frac{1}{2\pi i}\oint_{\Sp}^{}\frac{\partial_{1}\Cc\left( p,p \right)}{\Cc\left( p,p \right)}\mathrm{d}p\Bigg)N \label{eq1}.
\end{equation}
\end{proposition}
\begin{proof}
From \autoref{technicallemma1}, the pair $(K(p),K(q))$ has the same distribution as
\begin{equation*}
    \left(\sqrt{\Cc(q,q)} G_1,\frac{\Cc(p,q)G_1+\sqrt{\Cc(q,q)\Cc(p,p)-|\Cc(p,q)|^2}\,G_2}{\sqrt{\Cc(q,q)}}\right)
\end{equation*}
where $G_{1},G_{2}$ are independently drawn from the GinUE of size $N$. Substituting this representation into the definition of the partition function yields
\begin{align*}
        \mathscr{Z}_1^{\left(N\right)}\left(p,q\right)=\Cc(q,q)^{-N}\mathbb{E}\Big(\det(\Cc(p,q)I_N+\sqrt{\Cc(q,q)\Cc(p,p)-|\Cc(p,q)|^2}\,G_1^{-1}G_2)\Big).
\end{align*}
The distribution of the matrix $Y=G_{1}^{-1}G_{2}$ coincides with that of the complex spherical ensemble. In particular, its law is $\mathrm{U}(1)$-invariant in the following sense: for all $\theta\in[0,2\pi[$, $Y\overset{d}{=}e^{i\theta}Y$. This implies that only the scalar term $\Cc(p,q)$ contributes in the expectation, leading to the first identity \eqref{eq1}. The second identity follows directly by plugging~\eqref{eq1} into~\eqref{windingnumber1} and~\eqref{C1.def}. 
\end{proof}

\begin{remark}
    As mentioned previously, one recovers the Ginibre additive 2-matrix model by setting $C(p,q)=\nu(q)^{\dagger}\nu(p)$. In that case, a short computation shows that
    \begin{equation}
        \frac{\partial_{1}\Cc\left( p,p \right)}{\Cc\left( p,p \right)}=\frac{\nu(p)^{\dagger}\nu'(p)}{\nu(p)^{\dagger}\nu(p)},
    \end{equation}
    which agrees with the identity \eqref{asymptoticwinding1}. The Karhunen-Lo\`eve-type construction used in this subsection is, however, essentially restricted to the case $m=1$: for partition functions involving more than a single pair of points on the unit circle, a distributional identity of the form \eqref{K.Gauss} is no longer available. Furthermore, a comparison with the result~\eqref{windingasymptotic2} for the 2-matrix Muttalib-Borodin ensemble suggests that $\nu_{\gamma}\left( q \right)^{\dagger}\nu_{\gamma}\left( p \right)$ plays the role of an effective covariance kernel $C(p,q)$ beyond the Gaussian setting.
\end{remark}

\section{Partition function and proof of \autoref{theorem1}}\label{sec1}

The proof of \autoref{theorem1} mostly relies on standard complex analysis and determinantal identities. In fact, we extend the scope of our analysis to allow the possible capture of real and imaginary parts of our correlation functions by starting from~\eqref{partitionfunction2} under the conditions specified earlier for the coefficients.

We temporarily assume that for any indices $1\leq j_{1}\leq m_{1}$ and $1\leq j_{2}\leq m_{2}$ we have $b(p_{j_{1}})\neq 0$ and $b(\widetilde{p}_{j_{2}})\neq 0$, which can later be lifted by analytic continuation since the finite-$N$ result behaves regularly at these points.

 Following~\cite{hahnWindingNumberStatistics2023}, the partition function $\mathscr{Z}^{\left(N\right)}_{m_{1},m_{2}}$, see~\eqref{partitionfunction2}, can be recast into the form
\begin{align}\label{part}
    &\mathscr{Z}^{\left(N\right)}_{m_{1},m_{2}}\left(\textbf{p},\textbf{q},\widetilde{\textbf{p}},\widetilde{\textbf{q}}\right)=\Bigg( \prod\limits_{j=1}^{m_{1}}\frac{b\left ( p_{j} \right )}{b\left ( q_{j} \right )}\prod\limits_{\ell=1}^{m_{2}}\frac{\overline{b\left ( \widetilde{p}_{\ell} \right )}}{\overline{b\left ( \widetilde{q}_{\ell} \right )}} \Bigg)^{N} \nonumber
    \\
    \times&\mathbb{E}\Bigg(\prod_{j=1}^{m_{1}}\frac{\det\left ( \kappa\left ( p_{j} \right )I_{N}+K_{1}^{-1}K_{2} \right )}{\det\left ( \kappa\left ( q_{j} \right )I_{N}+K_{1}^{-1}K_{2} \right )}\prod_{\ell=1}^{m_{2}}\frac{\det\left ( \overline{\kappa\left ( \widetilde{p}_{\ell} \right )}I_{N}+\overline{K_{1}^{-1}K_{2}  } \right )}{\det\left ( \overline{\kappa\left ( \widetilde{q}_{\ell} \right )}I_{N}+\overline{K_{1}^{-1}K_{2}} \right )}\Bigg)
\end{align}
It is then possible to exploit the closure of P\'olya ensembles under matrix multiplication and inversion, thereby reducing the 2-matrix model to an effective 1-matrix model. 
We recall that if $K_{1}$ and $K_{2}$ are two random matrices independently drawn from a P\'olya ensemble with weight $\omega$ then the generalized ratio $K_{1}^{-1}K_{2}$ is drawn from a P\'olya ensemble with weight $\widehat{\omega}$ defined at \eqref{omegahat}. Furthermore, the random matrix $K_{1}^{-1}K_{2}$ admits the following jpdf for complex eigenvalues
    \begin{align}
        f^{(N)}_{\mathrm{EV}}\left(\mathbf{z}\right)=C^{\left ( N \right )}_{\mathrm{EV}}\left [ \widehat{\omega} \right ]\left | \Delta_{N}\left ( \mathbf{z} \right ) \right |^{2}\prod\limits_{j=1}^{N} \widehat{\omega}\left ( \left | z_{j} \right |^{2} \right ),
    \end{align}
see~\eqref{def:EVjpdf} and~\eqref{normalization1}.
The computation of the generalized partition function \eqref{partitionfunction2} therefore boils down to
\begin{align}\label{partitionfunction2.b}
    &\mathbb{E}\Bigg(\prod_{j=1}^{m_{1}}\frac{\det\left ( \kappa\left ( p_{j} \right )I_{N}+K_{1}^{-1}K_{2} \right )}{\det\left ( \kappa\left ( q_{j} \right )I_{N}+K_{1}^{-1}K_{2} \right )}\prod_{\ell=1}^{m_{2}}\frac{\det\left ( \overline{\kappa\left ( \widetilde{p}_{\ell} \right )}I_{N}+\overline{\left [K_{1}^{-1}K_{2}  \right ]} \right )}{\det\left ( \overline{\kappa\left ( \widetilde{q}_{\ell} \right )}I_{N}+\overline{\left [K_{1}^{-1}K_{2}  \right ]} \right )}\Bigg) \nonumber
    \\
    =&C^{\left ( N \right )}_{\mathrm{EV}}\left [ \widehat{\omega} \right ]\int_{\C^{N}}^{}\left | \Delta_{N}\left ( \mathbf{z} \right ) \right |^{2}\prod\limits_{k=1}^{N}\widehat{\omega}\left ( \left | z_{k} \right |^{2} \right )\prod\limits_{j=1}^{m_{1}}\Bigg(\frac{\kappa\left ( p_{j} \right )+z_{k} }{\kappa\left ( q_{j} \right ) +z_{k} }\Bigg)\prod\limits_{\ell=1}^{m_{2}}\Bigg(\frac{ \overline{\kappa\left ( \widetilde{p}_{\ell} \right )}+\overline{z_{k}} }{\overline{\kappa\left ( \widetilde{q}_{\ell} \right )} +\overline{z_{k}} }\Bigg)\mathrm{d}^{2}\mathbf{z},
\end{align}
with $\mathrm{d}^{2}\mathbf{z}=\prod\limits_{j=1}^{N}\mathrm{d}\Re\left( z_{j} \right)\mathrm{d}\Im\left( z_{j} \right)$.
Such integrals are known and have been computed in a general setting in~\cite[Eq.~(4.18)]{kieburgDerivationDeterminantalStructures2010}.

To keep the notation simple, we introduce a few more objects. One is a partition function in terms of an integral over the complex eigenvalues without the normalization constant,
\begin{align}\label{squaredVandermonde}
\begin{split}
    &\widetilde{\mathscr{Z}}_{\frac{n_{1}/m_{1}}{n_{2}/m_{2}}}^{\left(N\right)}\left(\textbf{p},\textbf{q},\widetilde{\textbf{p}},\widetilde{\textbf{q}}\right)\\
    =&\int_{\C^{N}}^{}\left | \Delta_{N}\left ( \mathbf{z} \right ) \right |^{2}\prod\limits_{j=1}^{N}\widehat{\omega}\left ( \left | z_{j} \right |^{2} \right )\frac{\prod\limits_{k_{1}=1}^{m_{1}}\Big( -\kappa\left( p_{k_{1}} \right) -z_{j} \Big)}{\prod\limits_{\ell_{1}=1}^{n_{1}}\Big( -\kappa\left( q_{\ell_{1}} \right)-z_{j} \Big)}\frac{\prod\limits_{k_{2}=1}^{m_{2}}\Big( -\overline{\kappa\left( \widetilde{p}_{k_{2}} \right)}-\overline{z_{j}} \Big)}{\prod\limits_{\ell_{2}=1}^{n_{2}}\Big( -\overline{\kappa\left( \widetilde{q}_{\ell_{2}} \right)} -\overline{z_{j}} \Big)}\mathrm{d}^{2}\mathbf{z},
\end{split} 
\end{align}
where $m_{1},m_{2},n_{1},n_{2}$ are four integers such that $n_{1}-m_{1}+N=n_{2}-m_{2}+N\geq 0$. Our goal is to compute this integral in the sub-case $m_{1}=n_{1}$ and $m_{2}=n_{2}$. We, however, define it in greater generality as kernels related to the partition functions $\widetilde{\mathscr{Z}}_{\frac{0/1}{0/1}}$ and $\widetilde{\mathscr{Z}}_{\frac{1/0}{1/0}}$ will appear later on. To that end, we also introduce the following auxiliary matrices which are expressed using these functions,
\begin{align}
    &\mathcal{J}^{\left(N\right)}_{1,1}\left ( \widetilde{\textbf{p}},\textbf{p} \right )=\Bigg(\widetilde{\mathscr{Z}}^{\left ( N \right )}_{\frac{0/1}{0/1}}\left ( \widetilde{p}_{i},p_{j} \right )\Bigg)_{\substack{1\leq i\leq m_{2}\\1\leq j\leq m_{1}}} \label{kernel1},
    \\
    &\mathcal{J}^{\left(N\right)}_{1,2}\left ( \textbf{p},\textbf{q} \right )=\Bigg(\frac{1}{\kappa\left(p_{i}\right)-\kappa\left(q_{j}\right)}\widetilde{\mathscr{Z}}^{\left ( N \right )}_{\frac{1/1}{0/0}}\left ( p_{i},q_{j} \right )\Bigg)_{1\leq i,j\leq m_{1}} \label{kernel2},
    \\
    &\mathcal{J}^{\left(N\right)}_{2,1}\left ( \widetilde{\textbf{p}},\widetilde{\textbf{q}} \right )=\Bigg(\frac{1}{\overline{\kappa\left(\widetilde{p}_{i}\right)}-\overline{\kappa\left(\widetilde{q}_{j}\right)}}\widetilde{\mathscr{Z}}^{\left ( N \right )}_{\frac{0/0}{1/1}}\left ( \widetilde{p}_{i},\widetilde{q}_{j} \right )\Bigg)_{1\leq i,j\leq m_{2}}  \label{kernel3},
    \\
    &\mathcal{J}^{\left(N\right)}_{2,2}\left ( \textbf{q},\widetilde{\textbf{q}} \right )=\Bigg(\widetilde{\mathscr{Z}}^{\left ( N \right )}_{\frac{1/0}{1/0}}\left ( q_{i},\widetilde{q}_{j} \right )\Bigg)_{\substack{1\leq i\leq m_{1}\\1\leq j\leq m_{2}}}  \label{kernel4}.
\end{align}
Then, the application of the identity~\cite[Eq.~(4.35)]{kieburgDerivationDeterminantalStructures2010} to \eqref{squaredVandermonde} yields the compact form
\begin{align}
    &\widetilde{\mathscr{Z}}_{\frac{m_{1}/m_{1}}{m_{2}/m_{2}}}^{\left(N\right)}\left(\textbf{p},\textbf{q},\widetilde{\textbf{p}},\widetilde{\textbf{q}}\right)\nonumber\\
    =&(-1)^{m_1m_2}C_{\mathrm{EV}}^{\left ( N \right )}\left [ \widehat{\omega} \right ] ^{m_1+m_2-1}\frac{\det\left(\begin{array}{c|c}
-N\mathcal{J}^{\left(N\right)}_{1,1}\left ( \widetilde{\textbf{p}},\textbf{p} \right ) & \mathcal{J}^{\left(N\right)}_{2,1}\left ( \widetilde{\textbf{p}},\widetilde{\textbf{q}} \right ) \\ \hline
\mathcal{J}^{\left(N\right)}_{1,2}\left ( \textbf{p},\textbf{q} \right ) & \frac{1}{N+1}\mathcal{J}^{\left(N\right)}_{2,2}\left ( \textbf{q},\widetilde{\textbf{q}} \right )
\end{array}\right)}{\det\left ( \frac{1}{\kappa\left ( p_{i_{1}} \right )-\kappa\left ( q_{j_{1}} \right )} \right )\det\left ( \frac{1}{\overline{\kappa\left ( \widetilde{p}_{i_{2}} \right )}-\overline{\kappa\left ( \widetilde{q}_{j_{2}} \right )}} \right )}.\label{part.nonnorm}
\end{align}
The first determinant in the denominator is of size $m_1\times m_1$ and the second one is of size $m_2\times m_2$.

We are now left with the task of computing explicitly three kernel functions since the entries in $\mathcal{J}_{1,2}$ and $\mathcal{J}_{2,1}$ have the same exact integral structure up to complex conjugation. Each of these kernels can be evaluated explicitly using a combination of determinantal identities of Vandermonde type, the generalized Andr\'eief identity \ref{th:andreiev} and classical tools from complex analysis such as contour integrals and residue calculus. Using the definition~\eqref{ypsilon.def}, the kernels~(\ref{kernel1}-\ref{kernel3}) take the compact forms
\begin{align}
    \widetilde{\mathscr{Z}}^{\left ( N \right )}_{\frac{0/1}{0/1}}\left ( p_{i},\widetilde{p}_{j} \right )=\widetilde{\mathscr{Z}}^{\left ( N \right )}_{\frac{0/1}{0/1}}\left ( \widetilde{p}_{j},p_{i} \right )=&\int_{\C^{N}}^{} \left | \Delta_{N}\left ( \mathbf{z} \right ) \right |^{2}\prod\limits_{\ell=1}^{N}\widehat{\omega}\left ( \left | z_{\ell} \right |^{2} \right )\left ( \kappa\left(p_{i}\right)+z_{\ell} \right )\overline{\left ( \kappa\left(\widetilde{p}_{j}\right)+z_{\ell} \right )}\mathrm{d}^{2}\mathbf{z} \nonumber
    \\
    =&\frac{\mathcal{M}[\widehat{\omega}](N+1)}{C^{\left( N \right)}_{\mathrm{EV}}\left[\widehat{\omega}\right]}\sum\limits_{\ell=0}^{N}\frac{\left ( \kappa\left(p_{i}\right)\overline{\kappa\left(\widetilde{p}_{j}\right)} \right )^{\ell}}{\mathcal{M}\left [ \widehat{\omega} \right ]\left ( \ell+1 \right )},\label{kernel3.a}
    \\
    \widetilde{\mathscr{Z}}^{\left ( N \right )}_{\frac{1/1}{0/0}}\left ( p_{i},q_{j} \right )=&\int_{\C^{N}}^{}\left | \Delta_{N}\left ( \mathbf{z} \right ) \right |^{2}\prod\limits_{\ell=1}^{N}\widehat{\omega}\left ( \left | z_{\ell} \right |^{2} \right )\Bigg(\frac{ \kappa\left(p_{i}\right)+z_{\ell} }{ \kappa\left(q_{j}\right)+z_{\ell} }\Bigg)\mathrm{d}^{2}\mathbf{z}\nonumber
    \\
    =&\frac{1}{C^{\left( N \right)}_{\mathrm{EV}}\left[ \widehat{\omega} \right]}\Big(1+\Big[1-\frac{\kappa\left(q_{j}\right)}{\kappa\left(p_{i}\right)}\Big]\Upsilon_{N}\left( \kappa\left( p_{i} \right),\kappa\left( q_{j} \right) \right)\Big),\label{kernel3.b}
    \\
    \widetilde{\mathscr{Z}}^{\left ( N \right )}_{\frac{1/0}{1/0}}\left ( q_{i},\widetilde{q}_{j} \right )=\widetilde{\mathscr{Z}}^{\left ( N \right )}_{\frac{1/0}{1/0}}\left ( \widetilde{q}_{j} ,q_{i}\right )=&\int_{\C^{N}}^{}\left | \Delta_{N}\left ( \mathbf{z} \right ) \right |^{2}\prod\limits_{\ell=1}^{N}\frac{\widehat{\omega}\left ( \left | z_{\ell} \right |^{2} \right )}{\left ( \kappa\left( q_{i} \right)+z_{\ell} \right )\overline{\left ( \kappa\left( \widetilde{q}_{j} \right)+z_{\ell} \right )}}\mathrm{d}^{2}\mathbf{z}\nonumber
    \\
    =&\frac{1}{\mathcal{M}[\widehat{\omega}](N)C^{\left( N \right)}_{\mathrm{EV}}\left[ \widehat{\omega} \right]}\Bigg[\mathcal{S}[\widehat{\omega}\sigma_{r,R}]\left(\kappa(q_{i})\overline{\kappa(\widetilde{q}_{j})}\right)\\
    &-\sum_{k=1}^{N-1}\frac{\mathcal{M}[\widehat{\omega}](k,\left|\kappa(q_{i})\right|^{2})\mathcal{M}[\widehat{\omega}](k,\left|\kappa(\widetilde{q}_{j})\right|^{2})}{\mathcal{M}[\widehat{\omega}](k)(\kappa(q_{i})\overline{\kappa(\widetilde{q}_{j})})^{k}}\Bigg],\label{kernel3.c}
\end{align}
with $r:=\min(\left|\kappa(q_{i})\right|^{2},\left|\kappa(\widetilde{q}_{j})\right|^{2})$, $R:=\max(\left|\kappa(q_{i})\right|^{2},\left|\kappa(\widetilde{q}_{j})\right|^{2})$. The function $\sigma_{r,R}$ along with the Stieltjes transform $\mathcal{S}$ are defined in Remark~\ref{stieltjes}.
 For clarity and brevity, detailed proofs of these formulas are deferred to Appendices~\ref{secA1} and~\ref{secA2}.

The kernel~\eqref{kernel3.a} is, actually, the pre-kernel of the $k$-point correlation function for the eigenvalues of the P\'olya ensemble associated to the weight function $\widehat{\omega}$, see \cite[Lemma 4.1]{kieburgExactRelationSingular2016}. In particular,
    \begin{align}
        \widetilde{\mathscr{Z}}^{\left ( N \right )}_{\frac{0/1}{0/1}}\left ( p_{i},\widetilde{p}_{j} \right )=\frac{\pi\mathcal{M}[\widehat{\omega}](N+1)}{ C^{\left( N \right)}_{\mathrm{EV}}[\widehat{\omega}]}\frac{K^{\left( N \right)}_{\mathrm{EV}}\left[ \widehat{\omega} \right]\left( \kappa(p_i),\kappa(\widetilde{p}_j) \right)}{\sqrt{\widehat{\omega}\left( \left| \kappa(p_i) \right|^{2} \right)\widehat{\omega}\left( \left| \kappa(\widetilde{p}_j) \right|^{2} \right)}},
    \end{align}
which reduces to a relation to the level density when $\widetilde{p}_j,p_i\to p$, i.e.,
    \begin{align}
        \widetilde{\mathscr{Z}}^{\left ( N \right )}_{\frac{0/1}{0/1}}\left ( p,p \right )=\frac{\pi\mathcal{M}[\widehat{\omega}](N+1)}{ C^{\left( N \right)}_{\mathrm{EV}}[\widehat{\omega}]}\frac{\rho^{\left( N \right)}_{\mathrm{EV}}\left[ \widehat{\omega} \right]\left(  \left| \kappa\left( p \right) \right|^{2} \right)}{\widehat{\omega}\left( \left| \kappa\left( p \right) \right|^{2} \right)}.
    \end{align}
    
We can now combine these partition function expressions with~\eqref{part} and~\eqref{part.nonnorm} into another main result of the present work. For that purpose, we need the matrices~\eqref{cauchykernel} and~\eqref{polyakernel} as well as the matrices
\begin{equation}
    \mathrm{R}^{\left ( N \right )}_{m_{1},m_{2}}\left [ \widehat{\omega} \right ]\left ( \widetilde{\mathbf{p}},\mathbf{p} \right ) =\left(N\mathcal{M}[\widehat{\omega}](N+1)\sum\limits_{\ell=0}^{N}\frac{\left ( a\left(p_{j}\right)\overline{a\left(\widetilde{p}_{i}\right)} \right )^{\ell}\left ( b\left(p_{j}\right)\overline{b\left(\widetilde{p}_{i}\right)} \right )^{N-\ell}}{\mathcal{M}\left [ \widehat{\omega} \right ]\left ( \ell+1 \right )}\right)_{\substack{1\leq i\leq m_2\\1\leq j\leq m_1}}
\end{equation}
and
\begin{equation}
    \mathrm{T}^{\left ( N \right )}_{m_{1},m_{2}}\left [ \widehat{\omega} \right ]\left ( \mathbf{q},\widetilde{\mathbf{q}} \right ) =\left(C_{\rm EV}^{(N)}[\widehat{\omega}]\frac{\widetilde{\mathscr{Z}}^{\left ( N \right )}_{\frac{1/0}{1/0}}\left ( q_{i},\widetilde{q}_{j} \right )}{(b(q_i)\overline{b(\widetilde{q}_j)})^N}\right)_{\substack{1\leq i\leq m_1\\1\leq j\leq m_2}}.
\end{equation}
Summarising the above derivation, we have proven the following more general theorem.

\begin{theorem}[Determinantal Structure of the Partition Function]\label{theorem3}
        Let $N,m_{1},m_{2}$ be three positive integers. The partition function \eqref{partitionfunction2} associated with an additive 2-matrix model where the two source matrices $K_{1}$ and $K_{2}$ are independently drawn from a Pólya ensemble of size $N$ with weight $\omega$ is given by
\begin{align}\label{maintheorem1}
            \mathscr{Z}^{\left(N\right)}_{m_{1},m_{2}}\left(\mathbf{p},\mathbf{q},\widetilde{\mathbf{p}},\widetilde{\mathbf{q}}\right)=(-1)^{m_1m_2}\frac{\det\left(\begin{array}{c|c}
            -\mathrm{R}^{\left ( N \right )}_{m_{2},m_{1}}\left [ \widehat{\omega} \right ]\left ( \widetilde{\mathbf{p}},\mathbf{p} \right ) & \overline{\widetilde{\mathrm{Q}}^{\left ( N \right )}_{m_{2}}\left [ \widehat{\omega} \right ]}\left ( \widetilde{\mathbf{p}},\widetilde{\mathbf{q}} \right ) \\ 
            \\\hline
            \overset{}{\widetilde{\mathrm{Q}}^{\left ( N \right )}_{m_{1}}\left [ \widehat{\omega} \right ]\left ( \mathbf{p},\mathbf{q} \right )} & \mathrm{T}^{\left ( N \right )}_{m_{1},m_{2}}\left [ \widehat{\omega} \right ]\left ( \mathbf{q},\widetilde{\mathbf{q}} \right )
            \end{array}\right)}{\det\Big ( \mathrm{Q}_{m_{1}}\left ( \mathbf{p},\mathbf{q} \right ) \Big )\det\Big ( \overline{\mathrm{Q}_{m_{2}}}\left ( \widetilde{\mathbf{p}},\widetilde{\mathbf{q}} \right ) \Big )},
\end{align}
where the Cauchy-like kernel $\mathrm{Q}_{m}$ has been defined at \eqref{cauchykernel}.
\end{theorem}

\autoref{theorem1} then follows directly by setting $m_{1}=m$ and $m_{2}=0$. Under the assumptions of Definition \eqref{partitionfunction2}, the three kernel functions are well-defined at all points $(q_{i},\widetilde{q}_{j})$ such that $b(q_{i})=0$ or $b(\widetilde q_{j})=0$, which guarantees the existence of the integral representation.

\section{Asymptotic expansion of the average winding number and proof of \autoref{theorem2}}\label{sec:asymp}

We now restrict the asymptotic analysis to the case where the source random matrices $K_{1}$ and $K_{2}$ are independently drawn from the Muttalib-Borodin ensemble of Laguerre type presented at \eqref{def:MB-weight}.

The associated P\'olya weight is
\begin{align}\label{MB.setup}
    \omega_{\mathrm{MB}}\left(t\right)=\frac{\gamma}{\Gamma((\delta+1)/\gamma)}t^{\delta}e^{-t^{\gamma}}\quad \text{with } \gamma>0 \text{ and } \delta>-1, 
\end{align}
where the parameters $\delta$ and $\gamma$ remain fixed as $N\to+\infty$ and we have set $\alpha=1$ since, as aforementioned, the global scaling factor drops out in this model.

In most of the literature, the parameter $\gamma$ is denoted by $\theta$ and the parameter $\delta$ is denoted by $\nu$. Here we deliberately use a different notation to avoid confusion with the complex-valued vector $\nu$ that appears throughout the paper and to emphasize that $\gamma$ plays the role of a tail parameter controlling the decay rate of the mean level density at infinity.

As we restrict ourselves to the sub-case $m_1=m$ and $m_2=0$ in \autoref{theorem3}, the asymptotic behaviour of $\Upsilon_{N}$ (defined in \eqref{ypsilon.def}) evaluated on the diagonal suffices. We focus on the first moment of the winding number, that is, in the derivative
\begin{equation}
    \Cc^{\left ( N \right )}_{1}\left ( p \right )=\left.\frac{\mathrm{d}}{\mathrm{d}p} \mathscr{Z}^{\left(N\right)}_{\mathrm{1}}\left(p,q\right) \right|_{q=p}=\frac{{\rm d}}{{\rm d}p}\left[\left(\frac{b(p)}{b(q)}\right)^N\biggl(1+\frac{\kappa(p)-\kappa(q)}{\kappa(p)}\Upsilon_N(\kappa(p),\kappa(q))\biggl)\right]_{q=p}.
\end{equation}
The derivative in the second term must be in prefactor $\kappa(p)=a(p)/b(p)$, otherwise the difference would vanish. Hence
\begin{equation}
\begin{split}\label{C1.expl.1}
    \Cc_1^{(N)}(p)=&N\frac{b'(p)}{b(p)}+\frac{\kappa'(p)}{\kappa(p)}\Upsilon_N(\kappa(p),\kappa(p))=N\frac{a'(p)}{a(p)}-\frac{\kappa'(p)}{\kappa(p)}\Upsilon_N\left(\frac{1}{\kappa(p)},\frac{1}{\kappa(p)}\right)
\end{split}
\end{equation}
where the second equality uses the symmetry relation \eqref{ypsilon.sym.2}, highlighting the interchange symmetry between $a(p)$ and $b(p)$.

Since $\Upsilon_{N}$ depends on $\kappa$ only through its modulus, we define
\begin{equation}\label{Xi.def}
    \Xi_N(A)=\frac{1}{N}\Upsilon_N(\sqrt{A},\sqrt{A})\,\text{ and }\,A=|\kappa(p)|^2=|a(p)/b(p)|^2.
\end{equation}
This normalization is for convenience: $\Xi_N(A)$ is an averaged quantity that remains of order one as $N$ goes to infinity.
\begin{remark}
    When interchanging differentiation in $p$ with integration or expectation, one must take care at points where $a(p)=0$ or $b(p)=0$. These cases never occur simultaneously by assumption \eqref{assumption1}. When one of the two coefficient functions vanishes the logarithmic derivative
\begin{equation}
    \left.\frac{\partial}{\partial p}\frac{\det K(p)}{\det K(q)}\right|_{q=p}=\left\{\begin{array}{cl} \displaystyle \frac{a'(p)}{b(p)}{\rm tr}(K_2^{-1}K_1) +N\frac{b'(p)}{b(p)}, & a(p)=0, \\ \displaystyle N\frac{a'(p)}{a(p)} +\frac{b'(p)}{a(p)}{\rm tr}(K_1^{-1}K_2), & b(p)=0,\end{array}\right.
\end{equation}
remains finite and analytic. The same holds if the expectation is taken before differentiation and in the limit $q\to p$, as seen from \eqref{C1.expl.1}. Moreover, $\Xi_N(|\kappa(p)|^2)=\Upsilon_N(\kappa(p),\kappa(p))=O(|a(p)|^{2\delta+2})$ when $a(p)\to0$ and similarly when $b(p)\to0$, ensuring pointwise convergence and integrability. Therefore, differentiation and expectation may safely be interchanged.
\end{remark}

Using previous equations \eqref{id:generalizedratioweight} and \eqref{Mellin-MBI} for Muttalib-Borodin self-inverse-convoluted P\'olya weights one obtains
\begin{equation}
    \Xi_N(A)=\frac{1}{N}\sum_{k=1}^{N}\frac{\mathcal{M}[\widehat{\omega}_{\mathrm{MB}}](k,A)}{\mathcal{M}[\widehat{\omega}_{\mathrm{MB}}](k)}=\frac{1}{N}\sum_{k=1}^{N}\varphi_{N}\left(\frac{2k-1}{2N},A\right),
\end{equation}
where $\varphi_{N}$ is defined as a Laplace-type integral
\begin{align}\label{phi.def}
    \varphi_{N}\left(\tau,A\right)=\frac{\int_{-\infty}^{\gamma\ln(A)}h\left( x \right)\exp\left( -N\,F_{\tau}\left( x \right)/\gamma \right)\mathrm{d}x}{\int_{-\infty}^{+\infty}h\left( x \right)\exp\left( -N\,F_{\tau}\left( x \right)/\gamma \right)\mathrm{d}x}=1-\varphi_{N}\left(1-\tau,A^{-1}\right)\in[0,1],
\end{align}
with
\begin{align}\label{notation}
     \tau=\frac{2k-1}{2N}\in]0,1[,\  h\left( x \right)=\left( 2\cosh(x/2)\right)^{-(2\delta+1)/\gamma},\  F_{\tau}\left( x \right)=\ln\left( 1+e^{ x} \right)-\tau x.
\end{align}
The substitution $t=e^{x/\gamma}$ in the Mellin integrals leads precisely to this form. The shift $k\mapsto k-\frac{1}{2}$ simply enforces symmetry about $\tau=1/2$, or equivalently about $k=(N+1)/2$ under the reflection $A\leftrightarrow A^{-1}$. In particular, the auxiliary functions satisfy $h(x)=h(-x)$, $F_\tau(x)=F_{1-\tau}(-x)$ and $\varphi_{N}(\tau,A)=1-\varphi_{N}(1-\tau,A^{-1})$.

For a given $\tau\in]0,1[$, the function $F_\tau$ has a unique critical point at
\begin{equation}
    x_*(\tau)=\ln\left(\frac{\tau}{1-\tau}\right) \qquad\Leftrightarrow\qquad \tau=\frac{1}{1+e^{-x_*(\tau)}}.
\end{equation}
When $x_{*}\left( \tau \right)>\gamma\ln(A)$, the main contribution arises from the upper boundary of the integration domain corresponding to
\begin{equation}\label{tau0.def}
    \tau_0=\frac{A^{\gamma}}{1+A^{\gamma}}\in]0,1[ \qquad\Leftrightarrow\qquad A=\left(\frac{\tau_0}{1-\tau_0}\right)^{1/\gamma}.
\end{equation}
The region where contributions remain unsuppressed has width of order $1/\sqrt{N\tau_0(1-\tau_0)}$ around this boundary.
Hence, the boundary regime $\min\{\tau_0,1-\tau_0\}=\mathcal{O}(N^{-1})$, or equivalently $\min\{A,A^{-1}\}=\mathcal{O}(N^{-1/\gamma})$ must be treated separately in \autoref{sec:zero}.

We therefore decompose the sum over $k$ into three parts:
\begin{equation}
    \sum_{k=1}^N=\underbrace{\sum_{k=1}^{\lfloor N(\tau_0-\Delta_N)+1/2\rfloor}}_{=S_-}+\underbrace{\sum_{k=\lfloor N(\tau_0-\Delta_N)+3/2\rfloor}^{\lceil N(\tau_0+\Delta_N)-1/2\rceil}}_{=S_0}+\underbrace{\sum_{k=\lceil N(\tau_0+\Delta_N)+1/2\rceil}^{N}}_{=S_+}
\end{equation}
where
\begin{align}\label{DeltaN.def}
 \Delta_N=N^{\varepsilon-1/2}\sqrt{\tau_0(1-\tau_0)}
\end{align}
for some fixed $\varepsilon>0$ such that $N^\varepsilon=o\left(\left(N\tau_0(1-\tau_0)\right)^{1/12}\right)$. Here $\lfloor \cdot\rfloor$ and $\lceil \cdot\rceil$ denote the floor and ceiling functions.

We assume that $N$ is large enough so that $\Delta_N<\min\{\tau_0,1-\tau_0\}$, ensuring that the summation ranges of $S_-$ and $S_+$ are non-empty. The condition \eqref{DeltaN.def} forces $\varepsilon$ to be sufficiently small so that the integration boundary $A$ lies within a neighbourhood of the critical point where the Gaussian approximation holds uniformly and the associated error terms stay negligible.

The sum $S_-$ can be related to $S_+$ through the identity
\begin{equation}\label{sum.trick}
\begin{split}
S_-=&\sum_{k=1}^{\lfloor N(\tau_0-\Delta_N)+1/2\rfloor}\varphi_{N}\left(\frac{2k-1}{2N},A\right)\\
=&\lfloor N(\tau_0-\Delta_N)+1/2\rfloor-\sum_{k=\lceil N(1-\tau_0-\Delta_N)+1/2\rceil}^N\varphi_{N}\left(\frac{2k-1}{2N},A^{-1}\right)
\end{split}
\end{equation}
which follows from the reflection property $\varphi_{N}(\tau,A)=1-\varphi_{N}(1-\tau,A^{-1})$. This relation expresses the lower tail $S_{-}$ in terms of the upper one evaluated at the inverted parameter $A^{-1}$.

\begin{itemize}
    \item The sum $S_+$ corresponds to the exponentially suppressed tail contributions and is further analysed in \autoref{sec:tails}.
    \item The expansion of the denominator integral in \eqref{phi.def} is presented in \autoref{sec:denom}.
    \item The main contribution arises from the central sum $S_{0}$: the denominator integral in \eqref{phi.def} admits a uniform Gaussian approximation in a neighbourhood of the boundary $A$, see \autoref{sec:edge}, while the discrete sum over $k$ is handled using an Euler-MacLaurin approximation scheme in \autoref{sec:sum}.
\end{itemize}

\subsection{Asymptotic behaviour of the tails}\label{sec:tails}

We now assume that $A=o(N^{1/\gamma})$. The leading contribution to the integral in the numerator of~\eqref{phi.def} for
\begin{equation}
    \tau=\frac{2k-1}{2N}>\frac{2\lceil N(\tau_0+\Delta_N)+1/2\rceil-1}{2N}>\tau_0
\end{equation}
comes from $x=\gamma\ln(A)$ since $F_\tau(x)$ is strictly decreasing for $x<\gamma\ln(A)=x_*(\tau_0)<x_*(\tau)$. Additionally, $F'(x)=1/(1+e^{-x})-\tau\leq F'(\gamma\ln(A))=\tau_0-\tau<0$ and $F'(x)$ is strictly increasing for $x\leq \gamma\ln(A)$, and $h(x)$ is bounded from above on $\R$ by $2^{-(2\delta+1)/\gamma}$.
Thus, the integrand is bounded from above by
\begin{equation}\label{bound}
\begin{split}
    h(x)\exp\left(\frac{N}{\gamma} (F(x_*(\tau))-F(x))\right)\leq&2^{-(2\delta+1)/\gamma}\exp\left(\frac{N}{\gamma} \int_{x}^{x_*(\tau)}F'(y){\rm d}y\right)\\
    \leq&2^{-(2\delta+1)/\gamma}\exp\left(-\frac{N}{\gamma}(\tau-\tau_0)(x_*(\tau)-x)\right).
\end{split}
\end{equation}
It follows that the corresponding summand satisfies
\begin{equation}
    \varphi_N(\tau,A)\underset{N\to \infty}{=}\mathcal{O}\left(\frac{1}{\sqrt{N}(\tau-\tau_0)}\exp\left[-\frac{N}{\gamma}(\tau-\tau_0)(x_*(\tau)-\gamma\ln(A))\right]\right).
\end{equation}
Since $\tau-\tau_0>\Delta_N/2$ and 
\begin{equation}
x_*(\tau)-\gamma\ln(A)=\int_{\tau_0}^\tau\frac{\mathrm{d}\rho}{\rho(1-\rho)}>\frac{\tau-\tau_0}{\max\{\tau(1-\tau),\tau_0(1-\tau_0)\}},
\end{equation}
the exponential factor behaves like
\begin{equation}
N(\tau-\tau_0)(x_*(\tau)-\gamma\ln(A))>\frac{N(\tau-\tau_0)^2}{\max\{\tau(1-\tau),\tau_0(1-\tau_0)\}}>\gamma c_2 N^{2\varepsilon},
\end{equation}
as in \eqref{DeltaN.def}. Hence, there exist constants $c_1,c_2>0$ such that the terms decay exponentially as $\mathcal{O}(N^{c_1}e^{-c_2 N^{2\varepsilon}})$.

Consequently, for the sums $S_{-}$ and $S_{+}$, by equation \eqref{sum.trick} and a similar argument under the transformation $(\tau_0,A)\to(1-\tau_0,A^{-1})$, we obtain
\begin{equation}\label{Spm.result}
S_+=\mathcal{O}(N^{1+c_1}e^{-c_2 N^{2\varepsilon}}) \quad{\rm and}\quad S_-=\lfloor N(\tau_0-\Delta_N)+1/2\rfloor-\mathcal{O}(N^{1+c_1}e^{-c_2 N^{2\varepsilon}}),
\end{equation}
where for $S_-$ we assume $A^{-1}\underset{N\to\infty}{=}o(N^{1/\gamma})$.

\subsection{Asymptotic expansion of the denominator}\label{sec:denom}

Again we assume $\max\{A,A^{-1}\}=o(N^{1/\gamma})$. A Laplace approximation naturally reveals the magnitude of the dominant contributions. For each $\tau\in]0,1[$, the function $F_{\tau}(\cdot)$ has a unique critical point
\begin{align}\label{crit.point}
    x_{*}\left(\tau\right)=\ln(\tau)-\ln(1-\tau)
\end{align}
 with 
\begin{equation}\label{2nd.der}
F''_{\tau}\left(x_{*}\right)=\frac{1}{4\cosh^2( x_*/2)}=(1-\tau) \tau>0
\end{equation}
and higher derivatives given by
\begin{equation}\label{F.expan}
\begin{split}
 F^{(3)}_{\tau}\left(x_{*}\right)=&(1-\tau)\tau(1-2\tau),\\
 F^{(4)}_{\tau}\left(x_{*}\right)=&(1-\tau)\tau[1-6(1-\tau)\tau],
 \end{split}
\end{equation}
For each integer $n\geq 2$, there exists a polynomial $P_{n}$ such that
\begin{equation}
    F^{(n)}_{\tau}(x_{*})=\tau(1-\tau)P_{n}(\tau).
\end{equation}
Consequently, there exist constants $c_{n}>0$ so that
\begin{equation}
    \left| F^{(n)}_{\tau}(x_{*})\right|\leq c_{n}\tau(1-\tau)\,\text{ for all }\tau\in[0,1].
\end{equation}

\noindent The map $\tau\mapsto x_{*}\left( \tau \right)$ is a strictly increasing $\mathscr{C}^{1}$-diffeomorphism from $\left] 0,1 \right[$ to $\R$.
Applying Laplace's method to the denominator of~\eqref{phi.def} gives
\begin{equation}\label{denom.expan}
\begin{split}
    \int_{-\infty}^{+\infty}h\left( x \right)\exp\left(-\frac{N}{\gamma}\,F_{\tau}\left( x \right) \right)\mathrm{d}x\underset{N\to \infty}{=}&\sqrt{\frac{2\pi\gamma}{N }}[(1-\tau)\tau]^{(2\delta+1-\gamma)/(2\gamma)}e^{-NF_\tau(x_*)}\\
    &\times\biggl[1+\frac{\chi(\tau,\delta,\gamma)}{24N\gamma(1-\tau)\tau }+\mathcal{O}\left(\frac{1}{N^2\tau^2(1-\tau)^2}\right)\biggl]
 \end{split}
\end{equation}
which is valid whenever $\tau$ and $1-\tau$ are much larger than $1/N$: for any fixed $0<\varepsilon<1$, the error term is uniform for $\tau\in[N^{\varepsilon-1},1-N^{\varepsilon-1}]$.

\noindent Here
\begin{equation}\label{chi.def}
\chi(\tau,\delta,\gamma)=12(1+\delta)\delta(1-2\tau)^2+2(12 \delta\gamma+6\gamma-\gamma^2-6)(1-\tau)\tau+3+2\gamma^2-6\gamma-12\delta\gamma
\end{equation}

For this expansion, we used
\begin{equation}\label{h.expan}
\begin{split}
h(x_*)=&[(1-\tau)\tau]^{(2\delta+1)/(2\gamma)},\\
h'(x_*)=&\frac{2\delta+1}{2\gamma}(1-2\tau)[(1-\tau)\tau]^{(2\delta+1)/(2\gamma)},\\
h''(x_*)=&\frac{2\delta+1}{4\gamma^2}[1+2\delta(1-2\tau)^2-4(1+\gamma)(1-\tau)\tau][(1-\tau)\tau]^{(2\delta+1)/(2\gamma)}.
\end{split}
\end{equation}
The regime $\tau=\mathcal{O}(1/N)$ can be neglected, since such terms only contribute to $S_{-}$ and $S_{+}$ both already handled in \autoref{sec:tails} under the assumption $\max\{A,A^{-1}\}\underset{N\to\infty}{=}o(N^{1/\gamma})$.

\subsection{Asymptotic expansion of the integral at the edge}\label{sec:edge}

We choose
\begin{equation}\label{edge.int}
\tau=\frac{2k-1}{2N}\in\left[\frac{2\lfloor N(\tau_0-\Delta_N)+3/2\rfloor-1}{2N},\frac{2\lceil N(\tau_0+\Delta_N)-1/2\rceil-1}{2N}\right]
\end{equation}
so that $|\tau-\tau_0|\underset{N\to \infty}{=}\mathcal{O}(\Delta_N)=o(1)$. 
Consequently,
\begin{equation}\label{x.dist}
|x_*(\tau)-\gamma\ln(A)|\underset{N\to \infty}{=}\mathcal{O}\left(\frac{\Delta_N}{\tau_0(1-\tau_0)}\right)=\mathcal{O}\left(\frac{N^\varepsilon}{\sqrt{N\tau_0(1-\tau_0)}}\right)
\end{equation}
since the distance $\left|x_*(\tau)-\gamma\ln(A)\right|$ depends on $N$ through $\left|\tau-\tau_{0}\right|$, and $\left|\tau-\tau_{0}\right|=\mathcal{O}(\Delta_{N})$.

Expanding the integrand about $x_*(\tau)$ over an interval of size
\begin{equation}
\Delta A_N=\frac{\Delta_N\ln(N\tau_0(1-\tau_0))}{\tau_0(1-\tau_0)}=\frac{N^\varepsilon\ln(N\tau_0(1-\tau_0))}{\sqrt{N\tau_0(1-\tau_0)}}
\end{equation}
shows that the leading contribution to the integral arises from the vicinity of the upper endpoint $x=\gamma\ln(A)$. The convergence of the expansion in the interval $[\gamma\ln(A)-\Delta A_N,\gamma\ln(A)]$ is uniform, since for all $n\geq 0$
\begin{equation}
\begin{split}
NF_\tau^{(n+3)}(x_*(\tau))\xi^{n+3}=&\mathcal{O}(N\tau(1-\tau)[\Delta A_N]^{n+3})\\
\underset{N\to\infty}{=}&\mathcal{O}(N^{(n+3)\varepsilon}[N\tau_0(1-\tau_0)]^{-(n+1)/2}[\ln(N\tau_0(1-\tau_0))]^{n+3})
\end{split}
\end{equation}
which indeed converges to zero. This follows from the choice of scaling \eqref{DeltaN.def} and from the limits $\Delta_N/\tau_0\to0$, $\Delta_N/(1-\tau_0)\to0$ and $N\tau_0\to\infty$.

We split the function $\varphi_{N}$ in two parts:
\begin{equation}
\varphi_N(\tau,A)=\varphi_N(\tau,e^{-\Delta A_N/\gamma}A)+\frac{\int_{\gamma\ln(A)-\Delta A_N}^{\gamma\ln(A)}h\left( x \right)\exp\left( -N\,F_{\tau}\left( x \right)/\gamma \right)\mathrm{d}x}{\int_{-\infty}^{+\infty}h\left( x \right)\exp\left( -N\,F_{\tau}\left( x \right)/\gamma \right)\mathrm{d}x}
\end{equation}

\noindent For $x<\gamma\ln(A)-\Delta A_N$, one has $x<x_*(\tau)$ with a distance $|x-x_*(\tau)|\gg\Delta_N/[\tau_0(1-\tau_0)]$, which allows us to bound the integrand uniformly similarly to \eqref{bound} with $\gamma \ln(A)$ replaced by $\gamma\ln(A)-\Delta A_N$. 

By setting
\begin{equation}
x_*(\widetilde\tau_N)=\gamma\ln(A)-\Delta A_N,
\end{equation}
we then obtain
\begin{equation}
\varphi_N(\tau,N^{-\Delta_N/\gamma}A)\underset{N\to\infty}{=}\mathcal{O}\left(\frac{1}{\sqrt{N}(\tau-\widetilde\tau_N)}\exp\left[-N(\tau-\widetilde\tau_N)\Delta A_N\right]\right).
\end{equation}
The right-hand side is exponentially small since $\tau-\widetilde\tau_N=(\tau-\tau_0)+(\tau_0-\widetilde\tau_N)$ with $\tau-\tau_{0}=\mathcal{O}(\Delta_N)$, while we estimate the second term by noting that $\Delta A_N=o(x_*(\tau_0))$ as $N$ goes to infinity which implies
\begin{equation}
\begin{split}
\tau_0-\widetilde\tau_N=&\int_{x_*(\tau_0)-\Delta A_N}^{x_*(\tau_0)}\frac{\mathrm{d}x}{4\cosh^2(x/2)}\underset{N\to\infty}{=}\mathcal{O}(\tau_0(1-\tau_0)\Delta A_N),
\end{split}
\end{equation}
since $1/[4\cosh^2(x_*(\tau_0)/2)]=\tau_0(1-\tau_0)$. Therefore, 
\begin{equation}
\varphi_N(\tau,N^{-\Delta_N/\gamma}A)\underset{N\to\infty}{=}\mathcal{O}(N^{c_{3}}e^{-c_{4} N^{2\varepsilon}[\ln(N\tau_0(1-\tau_0))]^2})
\end{equation}
for some positive constants $c_{3},c_{4}$.

We now expand the integrand around
\begin{equation}
    x=x_*(\tau)-\sqrt{\frac{\gamma}{N(1-\tau)\tau}}\xi
\end{equation}
over the interval $[\gamma\ln(A)-\Delta A_N,\gamma\ln(A)]$, using the earlier expansions \eqref{2nd.der}, \eqref{F.expan} and \eqref{h.expan}. The error remains uniform in this region.
The upper-limit in $\xi$,
\begin{equation}
    \sqrt{\frac{N(1-\tau)\tau}{\gamma}}[x_*(\tau)-\gamma\ln(A)+\Delta A_N]\underset{N\to\infty}{=} \mathcal{O}(N^\varepsilon\ln[N\tau_0(1-\tau_0)]),
\end{equation}
can be safely extended to $+\infty$ with an exponentially small error of order $\mathcal{O}(N^{c_{5}}e^{-c_{6}N^{2\varepsilon}[\ln(N\tau_0(1-\tau_0)]^2})$ for some positive constants $c_{5},c_{6}$.
However, the lower limit
\begin{equation}\label{lower.term}
   \xi_0=\sqrt{\frac{N(1-\tau)\tau}{\gamma}}[x_*(\tau)-\gamma\ln(A)]=\sqrt{\frac{N(1-\tau)\tau}{\gamma}}[x_*(\tau)-x_*(\tau_0)]
\end{equation}
must be retained, as it can be of order one or smaller.

After normalisation by the leading term from \eqref{denom.expan}, we obtain
\begin{equation}
\begin{split}
I=&\sqrt{\frac{N}{2\pi\gamma}}[(1-\tau)\tau]^{-(2\delta+1-\gamma)/(2\gamma)}\int_{\gamma\ln(A)-\Delta A_N}^{\gamma\ln(A)}h\left( x \right)\exp\left[ \frac{N}{\gamma}(F_{\tau}\left( x_*(\tau) \right)-F_{\tau}\left( x \right))\right]\mathrm{d}x\\
\underset{N\to\infty}{=}&\int_{\xi_0}^{+\infty}\biggl[1+\frac{C_{1,1}\xi+C_{1,3}\xi^3}{\sqrt{N\tau(1-\tau)}}+\frac{C_{2,2}\xi^2+C_{2,4}\xi^4+C_{2,6}\xi^6}{N\tau(1-\tau)}\biggl]e^{-\xi^2/2}\frac{\mathrm{d}\xi}{\sqrt{2\pi}}\\
&+\mathcal{O}\left(\frac{N^{9\varepsilon}[\ln(N\tau_0(1-\tau_0))]^9}{[N\tau_0(1-\tau_0)]^{3/2}}\right)
\end{split}
\end{equation}
where the coefficients are
\begin{equation}
\begin{split}
C_{1,1}=&\frac{(2\delta+1)(1-2\tau)}{2\sqrt{\gamma}},\quad C_{1,3}=\frac{\sqrt{\gamma}(2\tau-1)}{6},\\
C_{2,2}=&\frac{(2\delta+1)[1+2\delta(1-2\tau)^2-4(1+\gamma)(1-\tau)\tau]}{8\gamma},\\
C_{2,4}=&\frac{6\gamma(1-\tau)\tau-\gamma-2(2\delta+1)(1-2\tau)^2}{24},\\
C_{2,6}=&\frac{\gamma(1-2\tau)^2}{72}.
\end{split}
\end{equation}
They satisfy the identity
\begin{equation}
   C_{2,2}+3C_{2,4}+15C_{2,6}=\frac{\chi(\tau,\delta,\gamma)}{24\gamma}
\end{equation}
with $\chi$ defined in~\eqref{chi.def}.
Integrating by parts over $\xi$ and using the complementary error function $\mathrm {erfc}(x)=\frac{2}{\sqrt{\pi}}\int_x^\infty e^{-\xi^2}d\xi$, we obtain the asymptotic form
\begin{equation}\label{expand.phi.1}
\begin{split}
    \varphi_N(\tau,A)=&\frac{1}{2}\mathrm {erfc}\left(\frac{\xi_0}{\sqrt{2}}\right)+\biggl[C_{1,1}+C_{1,3}(\xi_0^2+2)+\mathcal{O}\left(\frac{\xi_0^5}{\sqrt{N\tau(1-\tau)}}\right)\biggl]\frac{e^{-\xi_0^2/2}}{\sqrt{2\pi N\tau(1-\tau)}}\\
    &+\mathcal{O}\left(\frac{N^{9\varepsilon}[\ln(N\tau_0(1-\tau_0))]^9}{[N\tau_0(1-\tau_0)]^{3/2}}\right).
\end{split}
\end{equation}
The prefactor of the complementary error function is exact, as it precisely matches the leading-order term in the denominator expansion. Errors are uniform in $\tau$ throughout the interval \eqref{edge.int}.

Since $|\tau_0-\tau|=\Delta_N\ll1$ we expand $\xi_0$, see~\eqref{lower.term}, as follows
\begin{equation}
\begin{split}
\xi_0\underset{N\to\infty}{=}&\sqrt{\frac{N}{\gamma(1-\tau_0)\tau_0}}(\tau-\tau_0)+\mathcal{O}\left(\frac{N^{3\varepsilon}}{N\tau_0(1-\tau_0)}\right)
\end{split}
\end{equation}
uniformly in $\tau$. Interestingly, the quadratic term $(\tau-\tau_0)^2$ cancels in this expansion.
Substituting into the previous result \eqref{expand.phi.1} yields the final asymptotic form:
\begin{equation}\label{expand.phi.2}
\begin{split}
    \varphi_N(\tau,A)\underset{N\to\infty}{=}&\frac{1}{2}\mathrm{erfc}\left(\sqrt{\frac{N}{2\gamma(1-\tau_0)\tau_0}}(\tau-\tau_0)\right)+\biggl[C_{1,1}+C_{1,3}\left(\frac{N(\tau-\tau_0)^2}{\gamma(1-\tau_0)\tau_0}+2\right)\biggl]\\
    &\times\frac{1}{\sqrt{2\pi N\tau_0(1-\tau_0)}}\exp\left(-\frac{N(\tau-\tau_0)^2}{2\gamma(1-\tau_0)\tau_0}\right)+\mathcal{O}\left(\frac{N^{5\varepsilon}}{N\tau_0(1-\tau_0)}\right),
\end{split}
\end{equation}
which will be used in the evaluation of the remaining sum $S_0$, where we used that $\xi_0=O(N^\varepsilon)$.

\subsection{Asymptotic expansion of the sum at the edge}\label{sec:sum}

The term involving the error function in \eqref{expand.phi.2} requires particular care. Once the discrete sum over $k$ is replaced by an integral in the scaled variable $t=\sqrt{N}(\tau-\tau_0)$ with $\tau=(2k-1)/2N$, it ceases to be integrable on the real line. To circumvent this issue, we decompose it as
\begin{equation}
\frac{1}{2}\mathrm{erfc}\left(\sqrt{\frac{N}{2\gamma(1-\tau_0)\tau_0}}(\tau-\tau_0)\right)=\mathbf{1}_{\left[0, \infty\right[}(\tau_{0}-\tau)+\frac{{\rm sign}(\tau-\tau_0)}{2}\mathrm{erfc}\left(\sqrt{\frac{N}{2\gamma(1-\tau_0)\tau_0}}|\tau-\tau_0|\right)
\end{equation}
where $\mathrm{sign}$ denotes the signum function.

The special case $N\tau_0+1/2\in\mathbb{Z}$ corresponds to an exceptional value of the continuous parameter $p$ and can therefore be ignored. We may thus assume that
\begin{equation}\label{kminus.def}
k_-=\left\lfloor N\tau_0+\frac{1}{2}\right\rfloor<k_+=\left\lceil N\tau_0+\frac{1}{2}\right\rceil.
\end{equation}
The same result extends to the integer case $k_{-}=k_{+}=k_{0}$, although one must then split the sum into two parts with end points $k_{0}=\pm 1$ and treat the central term $k=k_{0}$ separately.

We now decompose the sum $S_{0}$ as
\begin{equation}
\begin{split}
S_0=&\sum_{k=K_-}^{K_+}\varphi_N\left(\frac{2k-1}{2N},A\right)\\
=&k_--\left\lfloor N(\tau_0-\Delta_N)+\frac{1}{2}\right\rfloor+\left[\sum_{k=K_-}^{k_-}+\sum_{k=k_+}^{K_+}\right] f\left(\sqrt{N}\left[\frac{2k-1}{2N}-\tau_0\right]\right)\\
&+\mathcal{O}\left(N\Delta_N\frac{N^{5\varepsilon}}{N\tau_0(1-\tau_0)}\right)\text{ as } N \text{ goes to } \infty.
\end{split}
\end{equation}
with
\begin{equation}
    K_-=\lfloor N(\tau_0-\Delta_N)+3/2\rfloor,\,K_+=\lceil N(\tau_0+\Delta_N)-1/2\rceil
\end{equation}
and
\begin{equation}\label{f.def}
\begin{split}
f(t)=&\frac{{\rm sign}(t)}{2}\mathrm{erfc}\left(\frac{|t|}{\sqrt{2}}\right)+\frac{C_{1,1}+C_{1,3}\left(t^2+2\right)}{\sqrt{2\pi N\tau_0(1-\tau_0)}}\exp\left(-\frac{t^2}{2}\right)
\end{split}
\end{equation}
where we have set $t=\sqrt{N/[\gamma(1-\tau_0)\tau_0]}(\tau-\tau_0)$.
This splitting is essential, as the complementary error function is not differentiable at $\tau=\tau_0$, corresponding to a summation index lying between $k_{-}$ and $k_{+}$.

Using the definition of $\Delta_{N}$ \eqref{DeltaN.def}, the associated error term is 
\begin{equation}
    \mathcal{O}(N^{6\varepsilon}/\sqrt{N\tau_0(1-\tau_0)})
\end{equation}
which vanishes for our choice of $\varepsilon$.

We apply the Euler-MacLaurin formula separately for each of the two sums. For the first one we obtain
\begin{equation}
\begin{split}
&\sum_{k=K_-}^{k_-} f\left(\sqrt{\frac{N}{\gamma\tau_0(1-\tau_0)}}\left[\frac{2k-1}{2N}-\tau_0\right]\right)=\int_{K_-}^{k_-} f\left(\sqrt{\frac{N}{\gamma\tau_0(1-\tau_0)}}\left[\frac{2k-1}{2N}-\tau_0\right]\right)\mathrm{d}k\\
&+\frac{1}{2}\left[f\left(\sqrt{\frac{N}{\gamma\tau_0(1-\tau_0)}}\left[\frac{2k_--1}{2N}-\tau_0\right]\right)+f\left(\sqrt{\frac{N}{\gamma\tau_0(1-\tau_0)}}\left[\frac{2K_--1}{2N}-\tau_0\right]\right)\right]\\
&-\frac{1}{12\sqrt{N\tau_0(1-\tau_0)\gamma}}\\
&\times\left[f'\left(\sqrt{\frac{N}{\gamma\tau_0(1-\tau_0)}}\left[\frac{2k_--1}{2N}-\tau_0\right]\right)-f'\left(\sqrt{\frac{N}{\gamma\tau_0(1-\tau_0)}}\left[\frac{2K_--1}{2N}-\tau_0\right]\right)\right]\\
&+R,
\end{split}
\end{equation}
where the remainder is bounded by
\begin{equation}
\begin{split}
|R|\leq&\frac{2\zeta(2)}{N\tau_0(1-\tau_0)\gamma(2\pi)^2}\int_{K_-}^{k_-} \left|f''\left(\sqrt{\frac{N}{\gamma\tau_0(1-\tau_0)}}\left[\frac{2k-1}{2N}-\tau_0\right]\right)\right|\mathrm{d}k\\
\leq&\frac{2\zeta(2)}{\sqrt{N\tau_0(1-\tau_0)\gamma}(2\pi)^2}\int_{-\infty}^0\left|f''\left(t\right)\right|\mathrm{d}t\underset{N\to\infty}{=}\mathcal{O}\left(\frac{1}{\sqrt{N\tau_0(1-\tau_0)}}\right)
\end{split}
\end{equation}
where $\zeta$ denotes the Riemann zeta function $\zeta$, and we used the fact that
\begin{equation}\label{tpm.def}
t_-=\sqrt{\frac{N}{\gamma\tau_0(1-\tau_0)}}\left[\frac{2k_--1}{2N}-\tau_0\right]<0<\sqrt{\frac{N}{\gamma\tau_0(1-\tau_0)}}\left[\frac{2k_+-1}{2N}-\tau_0\right]=t_+.
\end{equation}
Both $f$ and $f'$ vanish exponentially as $k\to K_{-}$ since $t=O(N^\varepsilon)$. These boundary terms may therefore be neglected.

At $k=k_{-}$, one has $t=O(1/\sqrt{N\tau_{0}(1-\tau_{0})})$, the derivative term is suppressed by its prefactor whereas the value $f(k_{-})$ must be retained.
Furthermore, we replace $k$ by the integration variable $t$ appearing in the argument of $f$. The upper limit becomes $t_{-}$ while the lower limit is of order $N^{\varepsilon}$ and can therefore be extended to $-\infty$, since this modification only introduces an exponentially small error. Indeed, $f$ has a Gaussian decay near infinity.

This yields
\begin{equation}
\begin{split}
&\sum_{k=K_-}^{k_-} f\left(\sqrt{\frac{N}{\gamma\tau_0(1-\tau_0)}}\left[\frac{2k-1}{2N}-\tau_0\right]\right)\\
\underset{N\to\infty}{=}&\sqrt{N\tau_0(1-\tau_0)\gamma}\int_{-\infty}^{t_-} f\left(t\right)\mathrm{d}t+\frac{f(t_-)}{2}+\mathcal{O}\left(\frac{1}{\sqrt{N\tau_0(1-\tau_0)}}\right).
\end{split}
\end{equation}

A similar computation for the second sum gives
\begin{equation}
\begin{split}
S_0\underset{N\to\infty}{=}&k_--\left\lfloor N(\tau_0-\Delta_N)+\frac{3}{2}\right\rfloor+\sqrt{N\tau_0(1-\tau_0)\gamma}\left(\int_{-\infty}^{t_-} +\int_{t_+}^{+\infty}\right)f\left(t\right)\mathrm{d}t\\
&+\frac{f(t_-)+f(t_+)}{2}+\mathcal{O}\left(\frac{N^{6\varepsilon}}{\sqrt{N\tau_0(1-\tau_0)}}\right).
\end{split}
\end{equation}
Since $t_\pm=\mathcal{O}(1/\sqrt{N\tau_0(1-\tau_0)})$, we expand $f$ about $t=0$ using one-sided derivatives, as $f$ is discontinuous at the origin. In particular
\begin{equation}
    \lim_{t\to 0_{\pm}}f(t)=\pm\frac{1}{2}+\mathcal{O}\left(\frac{1}{\sqrt{N\tau_0(1-\tau_0)}}\right).
\end{equation}
This way, we obtain
\begin{equation}
\begin{split}
S_0\underset{N\to\infty}{=}&k_--\left\lfloor N(\tau_0-\Delta_N)+\frac{1}{2}\right\rfloor+\sqrt{N\tau_0(1-\tau_0)\gamma}\int_{-\infty}^{+\infty}f\left(t\right)\mathrm{d}t\\
&-\frac{\sqrt{N\tau_0(1-\tau_0)\gamma}(t_++t_-)}{2}+\mathcal{O}\left(\frac{N^{6\varepsilon}}{\sqrt{N\tau_0(1-\tau_0)}}\right).
\end{split}
\end{equation}
Using $k_{+}=k_{-}+1$ and the definitions of $t_{\pm}$, we have 
\begin{equation}
k_--\frac{\sqrt{N\tau_0(1-\tau_0)\gamma}(t_++t_-)}{2}=N\tau_0
\end{equation}
The integral involving the complementary error function in \eqref{f.def} vanishes identically, since the integrand is odd. The remaining Gaussian contribution can be evaluated explicitly, which gives
\begin{equation}
\begin{split}
S_0=&N\tau_0-\left\lfloor N(\tau_0-\Delta_N)+\frac{1}{2}\right\rfloor+\sqrt{\gamma}(C_{1,1}+3C_{1,3})+\mathcal{O}\left(\frac{N^{6\varepsilon}}{\sqrt{N\tau_0(1-\tau_0)}}\right).
\end{split}
\end{equation}

Combining all three contributions $S_{-}$, $S_0$ and $S_{+}$ we obtain
\begin{equation}
\Xi_N(A)=S_-+S_0+S_+=N\tau_0+\sqrt{\gamma}\left[C_{1,1}+3C_{1,3}\right]+o(1)
\end{equation}
or equivalently,
\begin{equation}
    \Xi_N(A)=\frac{N|a(p)|^{2\gamma}}{|a(p)|^{2\gamma}+|b(p)|^{2\gamma}}+\frac{\gamma-2\delta-1}{2}\frac{|a(p)|^{2\gamma}-|b(p)|^{2\gamma}}{|a(p)|^{2\gamma}+|b(p)|^{2\gamma}}+\mathcal{O}\left(\frac{N^{6\varepsilon}}{\sqrt{N\tau_0(1-\tau_0)}}\right).
\end{equation}
Substituting this result into the $1$-point correlation function \eqref{C1.expl.1} gives
\begin{equation}\label{C1-result.1}
\begin{split}
&\Cc_1^{(N)}(p)\underset{N\to\infty}{=}N\frac{a'(p)\overline{a(p)}|a(p)|^{2\gamma-2}+b'(p)\overline{b(p)}|b(p)|^{2\gamma-2}}{|a(p)|^{2\gamma}+|b(p)|^{2\gamma}}\\
&+\left(\frac{a'(p)}{a(p)}-\frac{b'(p)}{b(p)}\right)\biggl[\frac{\gamma-2\delta-1}{2}\frac{|a(p)|^{2\gamma}-|b(p)|^{2\gamma}}{|a(p)|^{2\gamma}+|b(p)|^{2\gamma}}+\mathcal{O}\left(\frac{N^{6\varepsilon}}{\sqrt{N\tau_0(p)(1-\tau_0(p))}}\right)\biggl].
\end{split}
\end{equation}
These asymptotics hold whenever $N\tau_0(1-\tau_0)\to\infty$, i.e. when $N^{1/(2\gamma)}|a(p)|\to\infty$ and $N^{1/(2\gamma)}|b(p)|\to\infty$ as $N\to\infty$.

When $p_0\in{\rm S}^1$ is a zero of order $m$ for either $a$ or $b$, we consider the symmetrised combination 
\begin{equation}
    \Cc_1^{(N)}(p_0e^{i\delta \varphi})+\Cc_1^{(N)}(p_0e^{-i\delta \varphi})
\end{equation}
which, as $N$ goes to infinity, expands as
\begin{equation}\label{Symmetrisation}
\begin{split}
&\Cc_1^{(N,0)}(p_0e^{i\delta \varphi})+\Cc_1^{(N,0)}(p_0e^{-i\delta \varphi})\\
+&\frac{imc_1-(m+1)c_2\delta\varphi}{ic_1\delta\varphi-c_2\delta\varphi^{2}}\biggl[\frac{\gamma-2\delta-1}{2}\frac{|c_1\delta\varphi^m|^{2\gamma}-|b(p_0)+ib'(p_0)\delta\varphi|^{2\gamma}}{|c_1\delta\varphi^m|^{2\gamma}+|b(p_0)+ib'(p_0)\delta\varphi|^{2\gamma}}+\mathcal{O}\left(\frac{N^{6\varepsilon}}{\sqrt{N|\delta\varphi|^{2m\gamma}}}\right)\biggl]\\
-&\frac{imc_1+(m+1)c_2\delta\varphi}{ic_1\delta\varphi+c_2\delta\varphi^{2}}\biggl[\frac{\gamma-2\delta-1}{2}\frac{|c_1\delta\varphi^m|^{2\gamma}-|b(p_0)-ib'(p_0)\delta\varphi|^{2\gamma}}{|c_1\delta\varphi^m|^{2\gamma}+|b(p_0)-ib'(p_0)\delta\varphi|^{2\gamma}}+\mathcal{O}\left(\frac{N^{6\varepsilon}}{\sqrt{N|\delta\varphi|^{2m\gamma}}}\right)\biggl]\\
=&\Cc_1^{(N,0)}(p_0e^{i\delta \varphi})+\Cc_1^{(N,0)}(p_0e^{-i\delta \varphi})-i\frac{(\gamma-2\delta-1)c_2}{2c_1}+\mathcal{O}\left(\frac{N^{6\varepsilon}}{\sqrt{N|\delta\varphi|^{2m\gamma}}}\right),
\end{split}
\end{equation}
where we expanded around $p_{0}$ such that $a(p_{0})=0$ under the assumption that there exists a strictly positive integer $m$ such that $a(p_0e^{i\delta \varphi})=c_1(i\delta\varphi)^m+c_2(i\delta\varphi)^{m+1}+\mathcal{O}(\delta\varphi^{m+2})$. The same reasoning applies if $b$ vanishes under identical conditions. We also denoted the leading order contribution as
\begin{equation}
\Cc_1^{(N,0)}(p)=N\frac{a'(p)\overline{a(p)}|a(p)|^{2\gamma-2}+b'(p)\overline{b(p)}|b(p)|^{2\gamma-2}}{|a(p)|^{2\gamma}+|b(p)|^{2\gamma}}.
\end{equation}

This symmetrisation removes the simple pole from the logarithmic derivatives of $a(p)$ and $b(p)$ provided that $|\delta \varphi|^{2\gamma m}\gg N^{-1}$. In the leading order, $\tau_{0}$ is an even function of $\delta\varphi$ as $\varphi\to 0$, which implies that $\tau_0=\mathcal{O}(|a(p)|^{2\gamma})=\mathcal{O}(|\delta\varphi|^{2\gamma m})$ when $a(p_0)=0$, and similarly, $1-\tau_0=\mathcal{O}(|b(p)|^{2\gamma})=\mathcal{O}(|\delta\varphi|^{2\gamma m})$ when $b(p)=0$. This symmetrisation underlies the appearance of the Cauchy principal value integral.

Finally, with this symmetrisation, the error term vanishes uniformly when the integration domain is restricted to a set of the form
\begin{equation}
{\rm S}^{1}_{\eta}=\{p\in {\rm S}^1|\,N\tau_0(p)(1-\tau_0(p))\geq N^{\eta}\}\,\text{ for }\eta\in]0,1].
\end{equation}
where 
\begin{equation}
\tau_0(p)(1-\tau_0(p))=\frac{|a(p)b(p)|^{2\gamma}}{(|a(p)|^{2\gamma}+|b(p)|^{2\gamma})^2}.
\end{equation}
The auxiliary exponent must satisfy $0<\varepsilon<\eta/12$.

In the next subsection, we will show that one can choose $\eta$ so that the remaining contribution from the complement ${\rm S}^{1}\setminus{\rm S}^{1}_{\eta}$ vanishes, except for the leading term in~\eqref{C1-result.1}, proportional to $N$, which reproduces the dominant contribution in \eqref{Symmetrisation}.

\subsection{Discussion for small parameter amplitudes}\label{sec:zero}

The previous analysis ceases to hold when $N\tau_0(1-\tau_0)$ is of order one. This regime corresponds to
\begin{equation}
    \min\left(NA^\gamma,NA^{-\gamma}\right)\underset{N\to\infty}{=}\mathcal{O}(1)
\end{equation}
or equivalently,
\begin{equation}
    \min\left(N^{1/(2\gamma)}|a(p)|,N^{1/(2\gamma)}|b(p)|\right)\underset{N\to\infty}{=}\mathcal{O}(1)
\end{equation}
In this situation, the discrete sum can no longer be approximated by an integral: the dominant contributions now arise from the terms near the boundaries $k=1$ or $k=N$. By symmetry, it suffices to discuss the case $NA^\gamma=\mathcal{O}(1)$. 

We assume that $N\tau_0(1-\tau_0)\leq N^{\eta}$ with $\eta\in]0,1]$. By definition of $\tau_{0}$, see \eqref{tau0.def}, it implies $A=\mathcal{O}(N^{-(1-\eta)/\gamma})$ and therefore $a(p)=\mathcal{O}(N^{-(1-\eta)/(2\gamma)})$.

Let $p_{0}\in\Sp$ be such that
\begin{equation}
    a(p_0e^{i\delta \varphi})=c_1(i\delta\varphi)^m+c_2(i\delta\varphi)^{m+1}+\mathcal{O}(\delta\varphi^{m+2})\,\text{ as }\delta\varphi\to 0.
\end{equation}
Then the size of the corresponding interval in $\delta\varphi$ is of order $\mathcal{O}(N^{-(1-\eta)/(2m\gamma)})$. Since the zeros of $a$ are isolated, one can choose $N$ large enough to ensure that the associated intervals around distinct zeros do not overlap.

We now split the sum $\Xi_N(A)$, see~\eqref{Xi.def}, into two parts:
\begin{equation}
    \sum_{k=1}^N=\underbrace{\sum_{k=1}^{\lceil N(\tau_0+\Delta_N)-1/2\rceil}}_{=\widetilde{S}_0}+\underbrace{\sum_{k=\lceil N(\tau_0+\Delta_N)+1/2\rceil}^{N}}_{=S_+}.
\end{equation}
The bound \eqref{Spm.result} for $S_+$ remains valid and shows that this term is exponentially small, uniformly for $N\tau_0\leq c N^{\eta}$ for some constant $c>0$.
Hence, only the first partial sum $\widetilde{S}_{0}$ contributes significantly, where the summation index satisfies $k=\mathcal{O}(N\Delta_N)=\mathcal{O}(N^{\varepsilon+\eta/2})$ as $N$ goes to infinity.

Using the explicit Mellin transform~\eqref{Mellin-MBI} and its corresponding weight \eqref{id:generalizedratioweight}, we write
\begin{equation}
\begin{split}
\varphi_{N}\left(\frac{2k-1}{2N},A\right)=&\frac{\mathcal{M}[\widehat{\omega}_{\mathrm{MB}}](k,A)}{\mathcal{M}[\widehat{\omega}_{\mathrm{MB}}](k)}\\
&\hspace*{-1cm}=\frac{\gamma\Gamma\left( (N + 2\delta + 1)/\gamma \right)}{\Gamma((\delta+k)/\gamma)\Gamma((N+\delta+1-k)/\gamma)}\int_0^A t^{\delta+k} \left(1 + t^\gamma \right)^{ -(N+2\delta+1)/\gamma}\frac{\mathrm{d}t}{t}.
\end{split}
\end{equation}
To apply Stirling's approximation formula to the Gamma functions, we require $k\ll N^{1/2}$. This is satisfied for any auxiliary exponents $0<\varepsilon<\eta/12<1/14$ provided that $\eta<6/7$.
Under this condition, 
\begin{equation}
 \frac{\Gamma\left( (N + 2\delta + 1)/\gamma \right)}{\Gamma((N+\delta+1-k)/\gamma)}=\left(\frac{N}{\gamma}\right)^{(k+\delta)/\gamma}\left[1+\mathcal{O}\left(N^{2\varepsilon+\eta-1}\right)\right]
\end{equation}
uniformly for all $k$ in the summation domain of $\widetilde{S}_0$.

For the integral, since $t^\gamma\leq A^\gamma=\mathcal{O}(N^{-(1-\eta)})$ as $N$ goes to infinity, we may further restrict $\eta<1/2$ so that uniformly
\begin{equation}
 \left(1 + t^\gamma \right)^{ -(N+2\delta+1)/\gamma}=\exp\left(-\frac{N}{\gamma} t^\gamma+\mathcal{O}(N^{-1+2\eta})\right).
\end{equation}
Substituting $t=(\gamma y/N)^{1/\gamma}$ with $y\in[0,1]$ gives
\begin{equation}
\begin{split}
\varphi_{N}\left(\frac{2k-1}{2N},A\right)=&\frac{1}{\Gamma((\delta+k)/\gamma)}\int_0^{NA^\gamma/\gamma} y^{(\delta+k)/\gamma}e^{- y}\frac{\mathrm{d}y}{y}\left[1+\mathcal{O}(N^{-1+2\eta})\right].
\end{split}
\end{equation}
The integrand attains its maximum at $y=(\delta+k-\gamma)/\gamma=\mathcal{O}(N^{\varepsilon+\eta/2})$, which indeed lies inside the integration domain as long as $\varepsilon+\eta/2<\eta$.
We therefore obtain the following rough upper-bound: 
\begin{equation}
\varphi_{N}\left(\frac{2k-1}{2N},A\right)\leq\frac{1}{\Gamma((\delta+k)/\gamma)}\left(\frac{\delta+k-\gamma}{\gamma e}\right)^{(\delta+k)/\gamma-1}\frac{NA^\gamma}{\gamma} \left[1+\mathcal{O}(N^{-1+2\eta})\right].
\end{equation}
For large $k$, the $k$-dependent part of this bound behaves like
\begin{equation}
\frac{1}{\Gamma((\delta+k)/\gamma)}\left(\frac{\delta+k-\gamma}{\gamma e}\right)^{(\delta+k)/\gamma}=\sqrt{\frac{\gamma}{2\pi(\delta+k)}}\left[1+\mathcal{O}\left(\frac{1}{k}\right)\right].
\end{equation}
Therefore, the sum $\widetilde{S}_0$ can be bounded from above as follows
\begin{equation}
\widetilde{S}_0=\sum_{k=1}^{\lceil N(\tau_0+\Delta_N)-1/2\rceil}\varphi_{N}\left(\frac{2k-1}{2N},A\right)\leq c\sqrt{N\Delta_N}\,NA^\gamma(1+o(1))=\mathcal{O}(N^{\varepsilon+5\eta/4})
\end{equation}
for some positive constant $c>0$. 
Substituting this bound into the symmetrised correlation function \eqref{Symmetrisation} yields
\begin{equation}
\begin{split}
\Cc_1^{(N)}(p_0e^{i\delta \varphi})+\Cc_1^{(N)}(p_0e^{-i\delta \varphi})=&\Cc_1^{(N,0)}(p_0e^{i\delta \varphi})+\Cc_1^{(N,0)}(p_0e^{-i\delta \varphi})\\
&\hspace*{-2cm}+\left[\frac{a'(p_0e^{i\delta \varphi})}{a(p_0e^{i\delta \varphi})}+\frac{a'(p_0e^{-i\delta \varphi})}{a(p_0e^{-i\delta \varphi})}\right]\mathcal{O}(N^{\varepsilon/2+5\eta/4})+\mathcal{O}(\delta\varphi).
\end{split}
\end{equation}
Multiplying the second term by the interval length which is of order $\mathcal{O}(N^{-(1-\eta)/(2m\gamma)})$, it indeed vanishes provided that
\begin{equation}
    0<\eta<\min\left(12/(31m\gamma+12),1/2\right),\,\text{ where }0<\varepsilon<\eta/12.
\end{equation}
Combining this result with the expansion obtained in \autoref{sec:sum}, and choosing $\eta$ accordingly, we decompose ${\rm S}^{1}\setminus{\rm S}^{1}_{\eta}=\Sigma_a\cup\Sigma_b$ where either $a$ or $b$ vanishes. The average winding number then, asymptotically in $N$, takes the form
\begin{equation}
\begin{split}
 \mathbb{E}\left(\mathrm{Wind}_{N} \right)=&\left[\int_{{\rm S}^{1}_{\eta}}+\int_{{\rm S}^{1}\setminus{\rm S}^{1}_{\eta}}\right]\Cc_1^{(N)}(p)\mathrm{d}p\\
 =&\int_{{\rm S}^{1}} \Cc_1^{(N,0)}(p)\mathrm{d}p+\frac{\gamma-2\delta-1}{2}\int_{{\rm S}^{1}_{\eta}} \left(\frac{a'(p)}{a(p)}-\frac{b'(p)}{b(p)}\right)\frac{|a(p)|^{2\gamma}-|b(p)|^{2\gamma}}{|a(p)|^{2\gamma}+|b(p)|^{2\gamma}}\mathrm{d}p\\
 &+o(1).
 \end{split}
 \end{equation}
 The integrals around the zeros in the second term are to be understood in the symmetrised sense, which in the limit $N\to\infty$ reduces to a Cauchy principal value integral.

\section{Conclusion}\label{sec13}

In this work we have investigated two natural extensions of the Gaussian additive 2-matrix model describing a disordered one-dimensional Hamiltonian in symmetry class AIII of the Altland-Zirnbauer classification.

First, we introduced a general complex Ginibre random field on the unit circle that recovers, as a special case, the standard Ginibre additive 2-matrix model while extending its analytical framework. To assess the universality of the Gaussian results, we then analysed a broader family of non-Gaussian models: P\'olya additive 2-matrix models at finite matrix size, focusing on the Muttalib-Borodin subclass in the large-$N$ regime. Compared to the GinUE, these random matrix ensembles exhibit different tail behaviour in their mean level density.

For both Gaussian and non-Gaussian settings, we derived explicit analytical expressions for the average winding number of the determinantal curve associated with the disordered bulk Hamiltonian as the quasi-momentum $p$ winds once around the unit circle modelling the one-dimensional Brillouin zone. We identified a geometric invariant in the leading term
\begin{align*}
    \mathbb{E}\Big(\mathrm{Wind}_{N}\Big)=\Bigg(\frac{1}{2\pi i}\oint_{{\rm S}^{1}}^{}\frac{\nu\left( p \right)^{\dagger}\nu'\left( p\right)}{\left\| \nu\left( p \right) \right\|^{2}}\mathrm{d}p\Bigg)N+\mathcal{O}\left(1\right).
\end{align*}
which essentially corresponds to an Aharonov-Anandan geometric phase angle. The subleading $\mathcal{O}\left(1\right)$ correction vanishes in the Gaussian case but generally persists for non-Gaussian P\'olya ensembles, where it can be interpreted as a re-weighted winding number around the origin. Its magnitude depends on the tail behaviour and repulsion from the origin in the mean level density, both of which determine the typical distance of the determinantal curve from the origin.

A central open question is the universality of the winding-number distribution itself. Under mild regularity conditions, one may expect a central limit theorem to hold, leading to an asymptotically Gaussian law as in the Ginibre additive 2-matrix model for which earlier derivations \cite{hahnUniversalCorrelations2024} show that the standard deviation scales as $N^{1/4}$. However, the precise variance and subleading corrections are expected to depend on ensemble-specific details as already highlighted for the mean winding number. This layered asymptotic structure is reminiscent of that observed in large-$N$ free-energy expansions of Coulomb gases, where the leading order captures universal mean-field behaviour, while coefficients in higher-order terms are of geometric nature and encode increasingly model-dependent information, see for example \cite[Th.~2]{SerfatyGaussianFluctuationsCoulomb2023}.

Extending this analysis to other symmetry classes such as CII and BDI in one dimension remains a challenging open direction. In these cases, the lack of real and quaternionic analogues of P\'olya ensembles of multiplicative type, combined with the Pfaffian rather than determinantal structure of correlation functions, renders the analysis substantially more involved, see for example \cite{hahnWindingNumberStatistics2023a}.

\backmatter

\bmhead{Acknowledgements}

MY acknowledges fruitful discussions with Matthias Allard on the topic of \textit{Pólya Ensembles on} $\mathrm{GL}_{N}(\C)$ and Sampad Lahiry on the topic of Laplace-type approximations. Both authors have been funded by the Australian Research Council via the grants DE210101581 (MY) and DP250102552 (MK).

\begin{appendices}

\section{Complex integral identities}\label{secA1}

In this appendix we present a set of three elementary complex integrals. These identities, while straightforward to derive using polar coordinates and residue calculus, will serve as key ingredients in the exact evaluation of the partition function.
 
\begin{proposition}
Let $k,\ell$ be two strictly positive integers, $f:\R_{+}\to\C$ be an integrable function whose moments up to the $(k-1)$st order exist, and let $ \alpha, \beta \in \C^* $ with $ \alpha \neq \beta $. Recalling the definition of the Mellin~\eqref{mellintransform1} and incomplete Mellin transform~\eqref{mellintransform1.inc}, the following identities hold
\begin{align}
\int_{\C} f(|z|^2)\, z^{k-1} \overline{z}^{\ell-1} \, \mathrm{d}^2 z 
&= \pi\, \mathcal{M}[f](k) \, \delta_{k,\ell}, \label{id:compint0} \\
\int_{\C} \frac{f(|z|^2)\, z^{k-1}}{\beta - \overline{z}} \, \mathrm{d}^2 z 
&= \frac{\pi}{\beta^k} \, \mathcal{M}[f](k, |\beta|^2), \label{id:compint1} \\
\int_{\C} \frac{f(|z|^2)}{(\alpha - z)(\beta - \overline{z})} \, \mathrm{d}^2 z 
&= \pi \left( \int_{0}^{\min(|\alpha|, |\beta|)^2} \frac{f(t)}{\alpha\beta - t} \, \mathrm{d}t
- \int_{\max(|\alpha|, |\beta|)^2}^{\infty} \frac{f(t)}{\alpha\beta - t} \, \mathrm{d}t \right). \label{id:compint2}
\end{align}
Here $\mathrm{d}^{2}z = \mathrm{d}\Re(z)\, \mathrm{d}\Im(z)$ denotes the Lebesgue measure on $\C$.
\end{proposition}
\begin{proof}
We begin with \eqref{id:compint0}. For $ \theta \in [0, 2\pi[ $, the substitution $ z \mapsto e^{i\theta} z $ shows that the integral vanishes unless $ k = \ell $. In the case $ k = \ell $, switching to polar coordinates $ z = \sqrt{r} e^{i\theta} $ yields
\begin{align}
    \int_{\C} f(|z|^2)\, |z|^{2(k-1)} \, \mathrm{d}^2 z = \frac{1}{2}\int_{0}^{2\pi} \mathrm{d}\theta \int_{0}^{\infty} f(r)\, r^{k-1} \, \mathrm{d}r = \pi \mathcal{M}[f](k),
\end{align}
as claimed.
 
Next we prove~\eqref{id:compint1} for which we use the parametrization $ z = \sqrt{r} e^{i\theta} $ and Fubini theorem so that
\begin{equation}\label{app.eq.A5}
    \int_{\C} \frac{f(|z|^2)\, z^{k-1}}{\beta - \overline{z}} \, \mathrm{d}^2 z = \frac{1}{2\beta}\int_{0}^{\infty}\mathrm{d}r f(r) r^{\frac{k-1}{2}} \fint_{0}^{2\pi} \frac{e^{ik\theta}}{e^{i\theta} - \sqrt{r}/\beta }\mathrm{d}\theta,
\end{equation}
where the Cauchy principal value integral, denoted by $\fint$, allows us to safely handle the simple pole located at $e^{i\vartheta}=\sqrt{r}/\beta$. A contribution from the angular integral occurs only when $r=|\beta|^2$, which forms a set of Lebesgue measure zero so we may restrict to $r\in]0,\infty[\setminus\{|\beta|^2\}$.

Fixing such $r$, we evaluate the angular integral by the residue theorem. It is non-zero only when $r<|\beta|^2$ in which case 
\begin{equation}
    \fint_{0}^{2\pi}\frac{e^{ik\theta}}{e^{i\theta} - \sqrt{r}/\beta }\mathrm{d}\theta=2\pi\left(\frac{r}{\beta^2}\right)^{(k-1)/2}.
\end{equation}
Substituting this into the previous expression yields
\begin{equation}
\int_{\C} \frac{f(|z|^2)\, z^{k-1}}{\beta - \overline{z}} \, \mathrm{d}^2 z =\frac{\pi}{\beta^k} \int_{0}^{|\beta|^2} f(r)\, r^{k-1} \, \mathrm{d}r = \frac{\pi}{\beta^k} \mathcal{M}[f](k, |\beta|^2).
\end{equation}
 
Finally we prove \eqref{id:compint2}. With the same parametrization $z = \sqrt{r} e^{i\theta}$ we obtain
\begin{equation}\label{app.eq.A8}
\int_{\C} \frac{f(|z|^2)}{(\alpha - z)(\beta - \overline{z})} \, \mathrm{d}^2 z = \frac{1}{2}\int_{0}^{\infty} \mathrm{d}r f(r) \fint_{0}^{2\pi} \frac{\mathrm{d}\theta}{(\alpha - \sqrt{r} e^{i\theta})(\beta - \sqrt{r} e^{-i\theta})}.
\end{equation}
The principal value once again ensures that the poles at $r=|\alpha|^2$ and $r=|\beta|^2$ are harmless as they form a set of Lebesgue measure zero.
By the residue theorem, the angular integral is non-zero precisely when $r<\min\left\{ \left| \alpha \right|^{2},\left| \beta \right|^{2} \right\}$ or $r>\max\left\{ \left| \alpha \right|^{2},\left| \beta \right|^{2} \right\}$, and in that case
\begin{equation}
    \fint_{0}^{2\pi} \frac{\mathrm{d}\theta}{(\alpha - \sqrt{r} e^{i\theta})(\beta - \sqrt{r} e^{-i\theta})}=\frac{2\pi}{\alpha\beta-r}.
\end{equation}
Substituting this into \eqref{app.eq.A8} gives the desired identity.
\end{proof}

\begin{remark}\label{stieltjes}
    The last identity \eqref{id:compint2} may be expressed as a Stieltjes transform,
\begin{equation}\label{stiel}
\int_{\C} \frac{f(|z|^2)}{(\alpha - z)(\beta - \overline{z})} \, \mathrm{d}^2 z 
= \pi\, \mathcal{S}[f  \sigma_{r,R}](\alpha\beta),
\end{equation}
with
\begin{equation}
\sigma_{r,R}\left(t\right) = \mathbf{1}_{[0, r]}\left(t\right) - \mathbf{1}_{\left[R, \infty\right[}(t), \quad r = \min(|\alpha|^2, |\beta|^2), \quad R = \max(|\alpha|^2, |\beta|^2).
\end{equation}
We recall that for any locally integrable function $ g : \R_+ \to \C $, the Stieltjes transform of $g$, denoted by $ \mathcal{S}[g] $, is defined as follows
\begin{equation}
\mathcal{S}[g](z) = \int_{\R_{+}}^{} \frac{g(t)}{z - t} \, \mathrm{d}t, \qquad z \in \C \setminus \R_+.
\end{equation}
\end{remark}

\section{Determinantal identities and derivation of the kernels~(\ref{kernel3.a}-\ref{kernel3.c})}\label{secA2}

This appendix gathers several elementary but non-trivial determinantal identities underlying the derivation of the kernels (\ref{kernel3.a}-\ref{kernel3.c}) and used repeatedly in the main text. While closely related to the classical Vandermonde and Andr\'eief identities, some of these formulas do not appear explicitly in standard references, so we include short derivations for completeness.

The evaluation of the three kernels~(\ref{kernel3.a}-\ref{kernel3.c}) relies on a generalized version of the Andr\'eief identity~\cite{andreevNoteRelationEntre1886} derived in~\cite[Appendix C.1]{kieburgDerivationDeterminantalStructures2010} and reviewed in \cite{forresterMeetAndreiefBordeaux2019}. For completeness we restate it here in a slightly more general form, allowing integration over an arbitrary measurable space.
 
\begin{proposition}[Generalized Andr\'eief Identity]
Let $ \left(E,\mu \right) $ be a measurable space, $N\in\mathbb{N}^*$ and $n,m\in\mathbb{N}$. Fix two matrices $A\in \mathrm{M}_{m, N+m}(\C)$ and  $B\in \mathrm{M}_{n, N+n}(\C)$ and let 
\begin{align*}
    f_{1},\ldots,f_{N+m}, g_{1},\ldots,g_{N+n}:E\to\C
\end{align*}
be measurable functions in $\mathrm{L}^2(E,\mu)$. Define
\begin{align*}
    F\left(\mathbf{x}\right):=\left(f_{j}\left(x_{i}\right)\right)\in\mathrm{M}_{N, N+m}(\C),\quad G\left(\mathbf{x}\right):=\left(g_{j}\left(x_{i}\right)\right)\in\mathrm{M}_{N, N+n}(\C).
\end{align*}
Then the following identity holds:
\begin{align}\label{th:andreiev}
    &\int_{E^{N}} \det\begin{pmatrix}
    A \\ F\left(\mathbf{x}\right)
    \end{pmatrix}
    \det\begin{pmatrix}
    B \\ G\left(\mathbf{x}\right)
    \end{pmatrix}
\mathrm{d}\mu\left(x_{1}\right)\cdots\mathrm{d}\mu\left(x_{N}\right)\nonumber
\\
=& \left(-1\right)^{mn} N!
 \det\left(\begin{array}{c|c}
0 & B \\ \hline
A^{\top} & \Big( \int_{E} f_{i}\left(x\right)g_{j}\left(x\right) \, \mathrm{d}\mu(x) \Big)_{\substack{1\leq i\leq N+m\\1\leq j\leq N+n}}
\end{array}\right).
\end{align}
\end{proposition}
Let $z_{1},\ldots,z_{N},\alpha,\beta$ be a collection of complex numbers and denote
\begin{align}
    \mathrm{V}_{m}\left( \textbf{z} \right)=\mathrm{V}_{m}\left( z_{1},\ldots,z_{N} \right)=\left( z^{j-1}_{i} \right)_{\substack{1\leq i\leq N\\1\leq j\leq m}}\in\mathrm{M}_{N,m}\left(\C\right),
    \\
    h\left(\textbf{z},\beta\right)=\Bigg(\frac{1}{\beta-z_{j}}\Bigg)_{1\leq j\leq N}\in\mathrm{M}_{N,1}\left(\C\right).
\end{align}
For the computation of kernel~\eqref{kernel3.a} we use the elementary identity
\begin{align}\label{eq:vandermonde}
    \Delta_{N+1}\left(\textbf{z},\alpha\right)=\Delta_{N}\left(\textbf{z}\right)\prod\limits_{j=1}^{N}\left( \alpha-z_{j} \right)
\end{align}
Applying this identity twice and combining it with the generalized Andr\'eief identity with $\mathrm{d}\mu\left(z_{i}\right)=f\left( \left| z_{i} \right|^{2} \right)\mathrm{d}\Re\left( z_{i} \right)\mathrm{d}\Im\left( z_{i} \right)$ gives
    \begin{equation}
    \begin{split}
        &\int_{\C^{N}}^{}\Delta_{N+1}\left( \textbf{z},\alpha \right)\Delta_{N+1}\left( \overline{\textbf{z}},\beta \right)\mathrm{d}\mu\left( \mathbf{z} \right)\\
        =&-N!\det\left(\begin{array}{c|c}
0 & \mathrm{V}_{N+1}\left( \beta \right) \\\hline
\mathrm{V}_{N+1}\left( \alpha \right)^{\top} & \left(\int_{\C}z^{i-1}\overline{z}^{j-1}\mathrm{d}\mu(z) \right)_{1\leq i,j\leq N+1}
\end{array}\right).
    \end{split}
    \end{equation}
The integral in the lower-right block, computed in~\eqref{id:compint0}, is diagonal. Using the Schur complement formula with $D$ an invertible matrix yields
    \begin{equation}
    \begin{split}
        &\int_{\C^{N}}^{}\Delta_{N+1}\left( \textbf{z},\alpha \right)\Delta_{N+1}\left( \overline{\textbf{z}},\beta \right)\mathrm{d}\mu\left( \mathbf{z} \right)=N!\left(\prod_{j=1}^{N+1}\pi\mathcal{M}[f](j)\right)\sum_{j=0}^N\frac{(\beta\alpha)^j}{\pi\mathcal{M}[f](j+1)},
    \end{split}
    \end{equation}
which gives~\eqref{kernel3.a}.

Next we assume that, for all $i$ in $[1,N]$, $\beta\neq z_{i}$. We apply the same identity in
\begin{equation}
\begin{split}
        &\int_{\C^{N}}^{} \left | \Delta_{N}\left ( \mathbf{z} \right ) \right |^{2}\prod\limits_{j=1}^{N}f\left ( \left | z_{j} \right |^{2} \right )\Bigg(\frac{ \alpha-z_{j} }{\beta-z_{j}}\Bigg)\mathrm{d}^{2}z_j\\
        =&(\beta-\alpha)\int_{\C^{N}}^{} \Delta_{N}\left ( \overline{\mathbf{z}} \right ) \det\left(\begin{array}{c|c}
\mathrm{V}_{N}\left( \textbf{z} \right)^{\top} & h\left( \textbf{z},\beta \right) \\ \hline
\mathrm{V}_{N}\left( \alpha \right)^{\top} & \displaystyle\overset{}{\frac{1}{\beta-\alpha}}
\end{array}\right)\prod\limits_{j=1}^{N}\mathrm{d}\mu(z_j)\\
=&N!(\beta-\alpha)\det\left(\begin{array}{c|c}
\left(\int_{\C}\overline{z}^{i-1}z^{j-1}\mathrm{d}\mu(z) \right)_{1\leq i,j\leq N} & \displaystyle\underset{}{\Bigl(\int_{\C}\frac{\overline{z}^{i-1}{\rm d}\mu(z)}{\beta-z}\Bigl)_{1\leq i\leq N}}\\ \hline
\mathrm{V}_{N}\left( \alpha \right)^{\top} & \displaystyle\overset{}{\frac{1}{\beta-\alpha}}
\end{array}\right).
\end{split}
\end{equation}
Here, the upper-left block is diagonal by~\eqref{id:compint0} and the upper-right block can be evaluated using~\eqref{id:compint1}:
\begin{equation}
    \int_{\C}\frac{\overline{z}^{i-1}{\rm d}\mu(z)}{\beta-z}=\frac{\pi}{\beta^i}\mathcal{M}[f](i,|\beta|^2).    
\end{equation}
Evaluating the determinant using the Schur complement formula gives
\begin{equation}
    \begin{split}
        &\int_{\C^{N}}^{} \left | \Delta_{N}\left ( \mathbf{z} \right ) \right |^{2}\prod\limits_{j=1}^{N}f\left ( \left | z_{j} \right |^{2} \right )\Bigg(\frac{ \alpha-z_{j} }{\beta-z_{j}}\Bigg)\mathrm{d}^{2}z_j\\
=&N!\left(\prod_{j=1}^{N}\pi\mathcal{M}[f](j)\right)\left[1-(\beta-\alpha)\sum_{j=1}^{N}\frac{\mathcal{M}[f](j,|\beta|^2)}{\mathcal{M}[f](j)}\frac{\alpha^{j-1}}{\beta^{j}}\right]
    \end{split}
    \end{equation}
which corresponds to~\eqref{kernel3.b}.

For the final formula~\eqref{kernel3.c}, we use
    \begin{align}
        \frac{\Delta_{N}\left(\mathbf{z}\right)}{\prod\limits_{j=1}^{N}\left(\alpha-z_{j}\right)}=\det\begin{pmatrix}
\mathrm{V}_{N-1}\left( \mathbf{z} \right) & h\left( \mathbf{z},\alpha \right)
\end{pmatrix}
    \end{align}
and the same identity for $\beta$ and the complex conjugate Vandermonde determinant to obtain
    \begin{equation}
    \begin{split}
        &\int_{\C^{N}}^{} \left | \Delta_{N}\left ( \mathbf{z} \right ) \right |^{2}\prod\limits_{j=1}^{N}\frac{f\left ( \left | z_{j} \right |^{2} \right )}{\left ( \alpha-z_{j} \right )\left ( \beta-\overline{z_{j}} \right )}\mathrm{d}^{2}z\\
        =&\int_{\C^{N}}^{}\det\begin{pmatrix}
\mathrm{V}_{N-1}\left( \mathbf{z} \right) & h\left( \mathbf{z},\alpha \right)
\end{pmatrix}\det\begin{pmatrix}
\mathrm{V}_{N-1}\left( \overline{\mathbf{z}} \right) & h\left( \overline{\mathbf{z}},\beta \right)
\end{pmatrix}\mathrm{d}\mu\left(\mathbf{z}\right)\\
=&N!\det\left(\begin{array}{c|c}
\left( \int_{\C}^{}f\left( \left| z \right|^{2} \right)z^{i-1}\overline{z}^{j-1}\mathrm{d}^{2}z \right)_{1\leq i,j\leq N-1} &\displaystyle \underset{ }{\left( \int_{\C}^{}\frac{f\left( \left| z \right|^{2} \right)z^{i-1}}{\beta-\overline{z}}\mathrm{d}^{2}z \right)_{1\leq i\leq N-1}} \\\hline \displaystyle
\left( \int_{\C}^{}\frac{f\left( \left| z \right|^{2} \right)\overline{z}^{j-1}}{\alpha-z}\mathrm{d}^{2}z \right)_{1\leq j\leq N-1} &\displaystyle \int_{\C}^{}\frac{f\left( \left| z \right|^{2} \right)}{\left( \alpha-z \right)\left( \beta-\overline{z} \right)}\mathrm{d}^{2}z 
\end{array}\right).
\end{split}
\end{equation}
Applying the generalized Andr\'eief identity \ref{th:andreiev} and evaluating the integrals via (\ref{id:compint0}-\ref{id:compint2}), the upper-left block is diagonal and the lower-right entry is the Stieltjes transform ~\eqref{stiel}. Applying the Schur complement formula once more yields
    \begin{equation}
    \begin{split}\label{id:vandermonde3}
    &\int_{\C^{N}}^{} \left | \Delta_{N}\left ( \mathbf{z} \right ) \right |^{2}\prod\limits_{j=1}^{N}\frac{f\left ( \left | z_{j} \right |^{2} \right )}{\left ( \alpha-z_{j} \right )\left ( \beta-\overline{z_{j}} \right )}\mathrm{d}^{2}z \nonumber
    \\
    =&\Bigg(\pi^{N}N!\prod\limits_{k=1}^{N-1}\mathcal{M}\left[ f \right]\left( k \right)\Bigg)\Bigg(\mathcal{S}\left[ f\sigma_{r,R} \right]\left( \alpha\beta \right)-\sum\limits_{k=1}^{N-1}\frac{\mathcal{M}\left[ f \right]\left( k,\left| \alpha \right|^{2} \right)\mathcal{M}\left[ f \right]\left( k,\left| \beta \right|^{2} \right)}{\left( \alpha\beta \right)^{k}\mathcal{M}\left[ f \right]\left( k \right)}\Bigg)
    \end{split}
    \end{equation}
which concludes the proof of~\eqref{kernel3.c}.

\end{appendices}

\nocite{*}
\bibliography{wdAIII}

\end{document}